\documentclass[aps,twocolumn,superscriptaddress,nofootinbib,a4paper,longbibliography]{revtex4-2}
\usepackage{graphicx, subfigure, float} % Required for inserting images
\usepackage{dcolumn}
\usepackage{bm}
\usepackage[colorlinks=true,linkcolor=blue,citecolor=magenta,urlcolor=blue]{hyperref}
\usepackage{physics}
\usepackage{amsthm}
\usepackage{bbold}
\usepackage{thm-restate}
\usepackage{amsmath, amssymb, amsfonts, amsthm, bm, bbold, physics}
\usepackage{mathtools}
\usepackage{xcolor, graphicx, subfigure, float, hyperref}
\usepackage{tcolorbox, changepage}
\usepackage{natbib}
\usepackage{dsfont}
\usepackage{thm-restate}
\usepackage[normalem]{ulem}
\usepackage{lineno}
\usepackage{tikz}
\usetikzlibrary{shapes.geometric}
\usepackage{subcaption}

\newtheorem{thm}{Theorem}
\newtheorem{lem}{Lemma}

\newtheorem{defi}{Definition}
\newtheorem{prop}{Proposition}

\def\>{\rangle}
\def\<{\langle}

\begin{document}
%\title{Quantum Nonlocality as a sufficient resource for Device-Independent Randomness Amplification of Santha-Vazirani sources}
\title{No Bound Randomness in Quantum Nonlocality}

\author{Ravishankar Ramanathan}
\email{ravi@cs.hku.hk}
\affiliation{School of Computing and Data Science, The University of Hong Kong, Pokfulam Road, Hong Kong}
\author{Yuan Liu}
\affiliation{School of Computing and Data Science, The University of Hong Kong, Pokfulam Road, Hong Kong}
\author{Yutian Wu}
\affiliation{School of Computing and Data Science, The University of Hong Kong, Pokfulam Road, Hong Kong}
\author{Stefano Pironio}
\affiliation{Laboratoire d'Information Quantique, CP224, Universit\'{e} libre de Bruxelles, 1050 Brussels, Belgium}

%%%%%%%%%%%%%%%%%%%%%%%%%%%%%%%%%%%%%%%%%%%%%%%%%%%%%%%%%%%%%%%%%%%

\date{\today}

%%%%%%%%%%%%%%%%%%%%%%%%%%%%%%%%%%%%%%%%%%%%%%%%%%%%%%%%%%%%%%%%%%%

\onecolumngrid

%%%%%%%%%%%%%%%%%%%%%%%%%%%%%%%%%%%%%%%%%%%%%%%%%%%%%%%%%%%%%%%%%%%

\begin{abstract}
Recently it has been found that there exist maximally nonlocal quantum correlations that fail to certify randomness for any fixed input pair, rendering them useless for device-independent spot-checking randomness expansion schemes.  Here we show that conversely, in DI randomness amplification protocols where all input pairs are used for randomness generation, any amount of quantum nonlocality is sufficient to certify randomness. This shows that no bound randomness exists in quantum nonlocality - any quantum nonlocal behavior is useful in a DI randomness generation task with appropriate modification of protocol structure. Secondly, we show that in contrast to the hitherto considered fixed-input guessing probability, the average guessing probability over all inputs is a faithful and monotonic measure of nonlocality. We use the average guessing probability to show that in contrast to findings in PRL 134, 090201, the detection efficiency threshold for randomness generation is never lower than that for nonlocality detection. Finally, we analytically compute the  average guessing probability by a quantum adversary of a single player's measurement outputs in a standard CHSH Bell test, and use it to demonstrate an improvement in the generation rate in state-of-art amplification protocols. 

\end{abstract}

\maketitle

\textit{Introduction.-}
The outcomes of local measurements on entangled quantum systems can be certified to be genuinely random through the violation of a Bell Inequality. The foundational phenomenon of Bell nonlocality has given rise to the possibility of Device-Independent (DI) quantum cryptography. In the DI framework, the honest users do not need to trust even the very devices used in the cryptographic protocol. They can instead directly verify the correctness and security by means of simple statistical tests of the device's input-output behavior, specifically checking for the violation of a Bell inequality. This violation would certify that the outputs of the device could not have been pre-determined, and as such could not be perfectly guessed by an adversary, allowing for the possibility of securely extracting randomness or secret key.

\textit{The setting.-} We consider the scenario in which Alice and Bob possess a pair of quantum devices or boxes that they can prevent from communicating (by shielding or separating them). The state of their devices is correlated with that of a device in possession of adversary Eve, so that the joint state of the three devices is described by a quantum state $\rho_{\mathcal{A}\mathcal{B}\mathcal{E}} \in \mathcal{S}(\mathcal{H}_{\mathcal{A}} \otimes \mathcal{H}_{\mathcal{B}} \otimes \mathcal{H}_{\mathcal{C}})$, where $\mathcal{S}(\mathcal{H})$ denotes the set of density operators in Hilbert space $\mathcal{H}$. Alice and Bob perform measurements $x \in X$ and $y \in Y$ on their respective devices with probabilities $\nu_{XY}(x,y)$, and obtain outcomes $a \in A$ and $b \in B$ respectively. This process is described by POVMs $\{M_{a|x}\}_x$ and $\{N_{b|y}\}_y$ acting on $\mathcal{H}_{\mathcal{A}}$ and $\mathcal{H}_{\mathcal{B}}$ respectively. Alice and Bob observe probabilities $p_{AB|XY}(a,b|x,y) = \text{Tr}\left[\rho_{\mathcal{A}\mathcal{B}} M_{a|x} \otimes N_{b|y} \right]$, where $\rho_{\mathcal{A}\mathcal{B}} = \text{Tr}_{\mathcal{E}} (\rho_{\mathcal{A}\mathcal{B}\mathcal{E}} )$. These $|A||B||X||Y|$ probabilities can be viewed as the components of a vector $\textbf{p} \in \mathbb{R}^{|A||B||X||Y|}$ which is usually called a (nonlocal) behavior $\textbf{p} = \{p_{AB|XY}\}$. For ease of notation, we will take $|A|=|B| = d$ and $|X| = |Y| = m$ throughout, as it doesn't affect the arguments in the paper. The state $\rho_{\mathcal{A}\mathcal{B}}$ and measurement operators $\{M_{a|x}\}_x$ and $\{N_{b|y}\}_y$ giving rise to behavior $\textbf{p}$ are known as a quantum realization $Q$ of $\textbf{p}$, a single behavior may admit several quantum realizations. The convex set of bipartite quantum behaviors is denoted by $\mathcal{Q}^{bip}$. 

In a device-independent protocol such as for randomness expansion \cite{PAM+10, BRC23, LLR+21, MS14} or randomness amplification \cite{FWE+23, KAF20, CR12, BRG+16, GMT+13, RBH+16, R23}, the adversary attempts to guess the outcomes of Alice and Bob. After Alice and Bob have performed measurements $x, y$ respectively, the correlations between their classical outcomes $a, b$ and the eavesdropper's system are described by the classical-quantum state $\sum_{a,b} p_{AB|XY}(a,b|x,y) |a,b \rangle \langle a,b| \otimes \rho_{\mathcal{E}}^{a,b,x,y}$ where $\rho_{\mathcal{E}}^{a,b,x,y}$ is the conditional state of Eve's system. Eve performs a measurement described by a $d^2$-element POVM $\{O_{e|z}\}$ in an attempt to guess $(a,b)$. 

Let us first consider the scenario of DI randomness expansion, where Alice and Bob attempt to expand a short random seed into a longer random string. DI protocols for randomness expansion are usually termed spot-checking \cite{BRC23, LLR+21, MS14}, in which each measurement round is declared to be either a test round (where the measurement settings are chosen according to some distribution $\nu(x,y)$) or a generation round (where the measurement settings are a fixed $(x^*, y^*)$), where test rounds occur with a small probability $\gamma$ (chosen using around $-\log_2 \gamma$ bits of randomness per round). The rationale behind a spot-checking protocol is that randomness is required to perform Bell tests, and it is desirable to run the protocol with a smaller amount of randomness. For instance, in the well-known CHSH inequality, if the inputs are chosen uniformly (with $\nu(x,y) = 1/4$ for all $x,y \in \{0,1\}$), then each test round requires $2$ bits of input randomness while the maximum amount of output randomness obtainable from the test is $1 + H_{\text{bin}}\left(\frac{1}{2} + \frac{1}{2\sqrt{2}}\right) \approx 1.6$ bits, where $H_{\text{bin}}$ denotes the binary entropy. In this scenario, a quantum eavesdropper attempts to guess Alice and Bob's outcomes for a fixed (but arbitrary) measurement setting $(x^*, y^*)$ given that Alice and Bob observe a behavior $\textbf{p}$. Eve's  guessing probability is then given by the solution to the following optimization problem~\cite{BSS14, NSPS14,BRC23,NSBSP18}
\begin{widetext}
\begin{eqnarray}
\label{eq:gpi-main}
   P^{(q)}_g(\textbf{p}, (x^*, y^*)) &=& \max_{\widetilde{\textbf{p}}^{a,b}} \sum_{a,b} \widetilde{p}^{a,b}_{AB|XY}(a,b|x^*,y^*) \nonumber \\
    &&\text{s.t.} \;\;\; \sum_{a,b} \widetilde{p}^{a,b}_{AB|XY}(a',b'|x',y') = p_{AB|XY}(a',b'|x',y') \; \; \; \; \forall a',b',x',y' \nonumber \\
    && \;\;\qquad \widetilde{\textbf{p}}^{a,b} \in \widetilde{\mathcal{Q}}^{bip}.
\end{eqnarray}
\end{widetext}
Here, $\widetilde{\mathcal{Q}}^{bip}$ denotes the set of non-normalised quantum behaviors, i.e., behaviors obtainable from measurements on non-normalised quantum states $\widetilde{\rho}_{\mathcal{A}\mathcal{B}}$ (i.e. with $\text{Tr}\left(\widetilde{\rho}_{\mathcal{A}\mathcal{B}} \right) \geq 0$). The quantity $P^{(q)}_g(\textbf{p}, (x^*, y^*))$ in \eqref{eq:gpi-main} is thus the relevant parameter for DI spot-checking randomness expansion protocols. We refer to Appendix \ref{sec:DIGP} for the details on the derivation of \eqref{eq:gpi-main} and other variants of the guessing probability. In \cite{RYS24}, some of us proved the following 
\begin{thm}[\cite{RYS24}]
There exists a quantum behavior $\textbf{p} \in \mathcal{Q}^{bip}$ that exhibits maximum nonlocality (achieves the maximum quantum value $\omega_{q}$ for some nonlocal game), but at the same time has the property that
\begin{eqnarray}
    P^{(q)}_g(\textbf{p}, (x^*, y^*)) = 1 \; \; \;\; \forall (x^*, y^*) \in |X| \times |Y|.
\end{eqnarray}
Therefore, even maximum quantum nonlocality is not always sufficient for spot-checking based DI randomness expansion.
\end{thm}

Now consider the scenario of DI randomness amplification, where Alice and Bob have access to a source of weakly random bits and their goal is to achieve a string of fully random bits. The weak sources may be considered to have the structure of a Santha-Vazirani (SV) source, wherein each bit produced by the source is biased up to $\epsilon \in [0,1/2)$ away from uniform when conditioned upon all previous bits, i.e., each bit $X_i$ produced by the $n$-bit source obeys
\begin{eqnarray}
    \frac{1}{2} - \epsilon \leq P(X_i | X_1, \ldots, X_{i-1}) \leq \frac{1}{2} + \epsilon \; \; \; \forall i = 1, \ldots, n.
\end{eqnarray} 
Alternatively the source may be considered to be a general min-entropy source, where an $(n,k)$ min-entropy source is an $n$-bit string with conditional (quantum) min-entropy $k$, meaning that the probability with which an adversary can guess the source correctly is $\leq 2^{-k}$ \cite{KRS09}. In protocols for DI randomness amplification \cite{FWE+23, KAF20, CR12, BRG+16, GMT+13, RBH+16, R23}, all runs are considered to both test runs and generation runs. That is, we consider the situation in which each input pair $(x,y)$ is chosen with some distribution $\mu_{XY}(x,y)$ and Eve attempt to guess Alice and Bob's outcomes on average over all settings $(x,y)$ given that Alice and Bob observe a behavior $\textbf{p}$. Eve's guessing probability in this case is given by the solution to the following optimization problem
\begin{widetext}
    \begin{eqnarray}
\label{eq:gpii-main}
 P^{(q)}_{g}(\textbf{p}, \mu_{XY}) &=& \max_{\{p_{ABE|XYZ}\}} \sum_{x,y} \mu_{XY}(x,y) \sum_{a,b} p_{ABE|XYZ}(a,b,e=(a,b)|x,y,z=(x,y)) \nonumber \\ 
   && \; \; \; \;\;\; \text{s.t.} \quad \sum_{e'} p_{ABE|XYZ}(a',b',e'|x',y',z') = {p}_{AB|XY}(a',b'|x',y') \; \; \; \; \forall a',b',x',y',z' \nonumber \\
   && \qquad \qquad \;  \{ p_{ABE|XYZ} \} \in \mathcal{Q}^{trip},
\end{eqnarray}
\end{widetext}
where it becomes necessary to consider the entire set of tripartite quantum correlations $\mathcal{Q}^{trip}$ for an eavesdropper with $m^2$ inputs and $d^2$ outcomes per input. The quantity $P^{(q)}_{g}(\textbf{p}, \mu_{XY})$ in \eqref{eq:gpii-main} is thus the relevant parameter for DI randomness amplification protocols. We now show the following
\begin{thm}
\label{thm:DIRA-main}
    Any quantum nonlocal behavior $\textbf{p} \in \mathcal{Q}^{bip}$ is a sufficient resource for DI randomness amplification of an $\epsilon$-Santha-Vazirani source for any $\epsilon \in [0,1/2)$. That is, for any nonlocal $\textbf{p} \in \mathcal{Q}^{bip}$, it holds that
    \begin{eqnarray}
        P^{(q)}_{g}(\textbf{p}, \mu_{XY}) < 1,
    \end{eqnarray}
    for inputs $x,y$ chosen using an $\epsilon$-SV source obeying $ (1/2 - \epsilon)^{\log_2 |X||Y|} \leq \mu_{XY}(x,y) \leq (1/2+ \epsilon)^{\log_2 |X||Y|}$.  
\end{thm}
Thm. \ref{thm:DIRA-main} implies that there is no ``bound randomness" in quantum nonlocal correlations, that is any behavior that exhibits even a small amount of quantum nonlocality is useful in the task of device-independent randomness amplification. The statements also imply that the quantity $P^{(q)}_{g}(\textbf{p}, \mu_{XY})$ for uniformly random $\mu_{XY}$ (or indeed any complete-support $\mu_{XY}$) is a faithful measure of quantum nonlocality in contrast to the quantity $P^{(q)}_g(\textbf{p}, (x^*, y^*))$ for any fixed $(x^*, y^*)$. 
We defer the proof of Thm. \ref{thm:DIRA-main} to Appendix \ref{app:no-bound-rand} which follows from the techniques used by Kempe et al. in \cite{KKMTV08} to show the NP-hardness of deciding whether the quantum value of three-player games is $1$ or $\leq 1 - \kappa$ for small $\kappa > 0$. We also show explicitly in Appendix \ref{app:pgavg-nl-meas} how the average guessing probability acts as a faithful measure of quantum nonlocality - in addition to being faithful, it is monotonic under operations that do not increase nonlocality (wiring and classical communication prior to the inputs~\cite{REAM12,DVJI14}). 

\textit{Detection efficiency thresholds for DI randomness extraction.-} We have seen in the previous sections that for any Bell inequality (with any number of players, inputs and outputs), its violation is a necessary and sufficient condition for device-independent randomness certification from the system, albeit in some cases one has to consider the average guessing probability and tailor the structure of the randomness certification protocol accordingly. In this section, we extend these considerations to the detection efficiency thresholds necessary for the certification of randomness. In particular, in a recent Letter \cite{ZLH+25}, the specific $I_{4422}$ inequality (with $X = Y = \{1,2,3,4\}$ and $A = B = \{0,1\}$) was considered. This inequality has been highlighted for its ability to tolerate a detection efficiency as low as $61.8\%$~\cite{vertesi2010closing} in contrast to the Eberhard inequality which achieves the optimal detection efficiency threshold of $2/3$ in the $2$ input, $2$ output scenario. Given this low detection efficiency threshold for detecting nonlocality, this inequality is considered to be of high practical interest for experimental implementation of DI randomness expansion or amplification protocols. However, in \cite{ZLH+25}, it was shown by means of a numerical calculation that the fixed setting min-entropy $H_{\min}(A,B|x^*,y^*,E)$ falls to zero for any fixed $(x^*, y^*)$ for a detection efficiency below $90.6\%$ thereby calling into question the utility of the $I_{4422}$ in practical DI implementations.

In Appendix \ref{app:det-eff-I4422}, we show by means of an explicit calculation that a consideration of $H_{\min}(A,B|X,Y,E)$ (obtained by computing  $-\log_2 P^{(q)}_{g}(\textbf{p}, \mu_{XY})$ for $\mu_{XY}(x,y) = 1/4$ for $x = y \in \{1,2,3,4\}$) in fact allows to certify randomness for any detection efficiency above $61.8\%$. In particular, we show how to extract optimal randomness from the experimental violation of $I_{4422}$ in \cite{ZLH+25}, also remarking that the inequality should be considered an ideal candidate for amplification and even for expansion protocols, provided we tailor the structure of the protocol to be one without spot-checking.

\textit{Analytical computation of the average guessing probability.-} The determination of the guessing probability by numerical means is essentially a solved problem as we have seen. In contrast, analytical tight bounds on the guessing probability have only been computed in certain specific cases \cite{PAM+10, KW16, WBA18,WAP21}. 

Consider the most well-known case of the CHSH Bell test in which Alice measures binary observables $A_0, A_1 \in \{ \pm 1\}$ and Bob measures similarly binary observables $B_0, B_1 \in \{ \pm 1\}$ chosen with uniform probabilities in the test rounds. Suppose that Alice and Bob observe a value $I_{obs} \in [2, 2\sqrt{2}]$ for the CHSH Bell expression given as
$I_{CHSH} = \langle A_0 B_0 \rangle + \langle A_0 B_1 \rangle + \langle A_1 B_0 \rangle - \langle A_1 B_1 \rangle$, with the correlators defined as $\langle A_i B_j \rangle:= \sum _{a,b\in\{\pm 1\}} ab\cdot p_{AB|XY}(a,b|A_i,B_j).$
% \begin{eqnarray}
%     \langle A_i B_j \rangle = p_{AB|XY}(+1,+1|A_i,B_j) + p_{AB|XY}(-1,-1|A_i,B_j) \nonumber \\
%     - p_{AB|XY}(+1,-1|A_i,B_j) - p_{AB|XY}(-1,+1|A_i,B_j). \; \; \; \nonumber
% \end{eqnarray}

Suppose now that during the randomness generation rounds Alice chooses her inputs $A_0$ and $A_1$ with probabilities $p$ and $1-p$ respectively for some $p \in \left[ \frac{1}{2}, 1 \right]$. We are then interested in the analytical form of the average guessing probability $P_g^{(q, A)}(I_{obs}, \{p,1-p\})$ of Alice's measurement outputs by a quantum adversary Eve. This local average guessing probability is defined as
\begin{eqnarray}
\label{eq:gp-al-Bell}
 &&P^{(q, A)}_{g}(I_{obs}, \{p,1-p\}) \nonumber \\&=& \max_{\{p_{ABE|XYZ}\}} \sum_{a= \pm 1} p \cdot p_{AE|XZ}(a, e=a|A_0, z=0) \nonumber \\
 &&\qquad \qquad \quad + (1-p) \cdot p_{AE|XZ}(a, e=a|A_1, z= 1) \nonumber \\ 
   && \qquad \qquad \;\; I_{CHSH}\left(p_{AB|XY} \right) = I_{obs} \; \; \; \; \nonumber \\
   && \; \qquad \qquad \;  \{ p_{ABE|XYZ} \} \in \mathcal{Q}^{trip}.
\end{eqnarray}
Here, $p_{AE|XZ}$ and $p_{AB|XY}$ are marginals of the tripartite quantum behavior $p_{ABE|XYZ}$.

Famously in \cite{PAM+10}, a tight bound on the fixed setting guessing probability for an observed value $I_{obs}$ of the CHSH Bell expression was derived as 
\begin{equation}
\label{eq:CHSH-guessprob}
    P^{(q, A)}_{g}(I_{obs}, \{1,0\}) = \frac{1}{2} + \frac{1}{2} \sqrt{2 - \frac{I_{obs}^2}{4}}.
\end{equation}
In the Appendix \ref{app:Pgavg-CHSH}, we generalise this well-known bound to the scenario where Alice chooses inputs with probabilities $p, 1-p$ in the generation rounds and obtain the following tight bound
\begin{equation}
    P^{(q,A)}_{g}(I_{obs},\{p,1-p\})
    = \frac12 + \frac{\sqrt{1-2p+2p^2}}{2} \, f_{\phi_p}(I_{obs}),
\end{equation}
where $\phi_p := 2\arctan\!\left(\frac{1-p}{p}\right)$ and
$f_{\phi}(I_{obs}):=\frac{\sqrt{\frac{\left(4-\alpha_{\min,\phi}^2\right)\left(32-2\alpha_{\min,\phi}^4\cos^2 \phi\right)}{16-8 \alpha_{\min,\phi}^2+\alpha_{\min,\phi}^4\cos^2 \phi}}-I_{obs}}{\alpha_{\min,\phi}}$. Here $\alpha_{\min,\phi}$ is the real root in $[0,\alpha_\phi]$ of
$h_{\phi}(t) = I_{obs}$, with $h_{\phi}(t) :=
\frac{4\sqrt{2}\,[32(8-8t^2+t^4)+2t^4\cos^2\phi(32-8t^2+t^4)-t^8\cos^4\phi]}
{\sqrt{(t^2-4)(t^4\cos^2\phi-16)}\,(16-8t^2+t^4\cos^2\phi)^{3/2}}$, and $\alpha_{\phi}$ the real solution in $[0,2]$ of $\frac{\alpha^2 \sin \phi (96-16\alpha^2-2\alpha^4\cos^2\phi)}
{(4-\alpha^2)(32-16\alpha^2+2\alpha^4\cos^2\phi)} = 1$.

Specifically, in the limiting case when $p=1$ we have that $\phi_{p=1} = 0$. Then for any $I_{obs}\in[2,2\sqrt{2}]$, we calculate
\begin{equation}
    \begin{split}
        \alpha_{\min, \phi_{p=1}}=\frac{2\sqrt{8-I_{obs}^2}}{I_{obs}}, \qquad I_Q^{(\alpha_{\min,\phi_{p=1}},\phi_{p=1})}=\frac{8}{I_{obs}}.
    \end{split}
\end{equation}
We thus obtain the function $f_{\phi_{p=1}}$ as 
\begin{eqnarray}
    f_{\phi_{p=1}}=\frac{\sqrt{8-I_{obs}^2}}{2},\; \; 
 %   P^{(q, A)}_{g}(I_{obs}, \{1,0\}) = \frac{1}{2} + \frac{1}{2} \sqrt{2-\frac{I^2_{obs}}{4}}.\nonumber \\
\end{eqnarray}
and the corresponding guessing probability in the limiting case of $p=1$ as \eqref{eq:CHSH-guessprob}.

In Appendix \ref{app:DIRA-imp}, we show how a consideration of the average guessing probability for a measurement-dependent locality (MDL) inequality based on the Hardy paradox \cite{RHAP+18} leads to an improvement in the generation rates of state-of-art DI randomness amplification protocols of Santha-Vazirani sources against quantum adversaries \cite{KAF20}. Specifically, in this Appendix we compute a lower bound on the conditional von Neumann entropy $H(A|X,E)$ of Alice's measurement outcomes conditioned on Eve's quantum side information under the constraint that Alice's input is chosen with an $\epsilon$-SV source, i.e., that $\frac{1}{2} - \epsilon \leq \mu_X(x) \leq \frac{1}{2} + \epsilon$. This quantity naturally appears in randomness amplification protocols, yet all the DIRA protocols thus far still made use of the fixed-setting guessing probability. By explicit computation of this average conditional von Neumann entropy via the SDP hierarchy introduced by Brown, Fawzi and Fawzi~\cite{BFF24}, we demonstrate an improvement in the generation rates for different $\epsilon$ and noise in the protocol. Finally, in Appendix \ref{app:Pgavg-simplex}, we provide analytical computations for an optimally deigned Bell expression which provides a larger gap and highglights the difference between considerations of the fixed-setting and the average guessing probability.

\textit{Conclusions and Open Questions.-} In this paper, we have shown by a careful consideration of different device-independent protocol structures that quantum nonlocality is a necessary and sufficient resource for DI randomness against quantum adversaries in protocol structures corresponding to randomness amplification. This complements previous findings that even maximum quantum nonlocality is not always sufficient for protocol structures corresponding to randomness expansion. This explicitly clarifies the relationship of the resource of quantum nonlocality vis-a-vis the resource of DI randomness.  Our considerations have also found a practical application in an explicit improvement in the generation rates of state-of-art randomness amplification protocols, as well as an explicit analytical computation of the average min-entropy in the CHSH Bell test which provides potential improvements in state-of-art expansion protocols. As a further consequence, we have shown that incorporating the average guessing probability provides an improvement in the detection efficiency tolerance in randomness extraction from inequalities such as $I_{4422}$.

Several open questions remain. A significant one is developing DI randomness expansion protocols beyond the spot-checking structure with improvements accrued by incorporating the average guessing probability over the fixed-setting case. Another question is to identify extremal quantum correlations that certify high average min-entropy for use in such protocols, in comparison to a number of recent findings where extremal correlations have found that certify $\log_2 d$ bits of randomness in fixed settings for local dimension $d$. Finally, it would also be interesting to investigate a resource theory of DI randomness and key and their relationship with the resource theory of quantum nonlocality. 

%We have seen that considerations of  average min-entropy against a quantum adversary can give rise to more randomness than that for a fixed setting. In this section, we introduce a bipartite Bell inequality in the  which certifies the maximum amount of 

Additional references~\cite{ACP+16, MS17, MS17-2, Win99, NPA1, NPA2, FBL+21, mikos2023extremal, SV84, Hardy93, Helstrom69} are cited in the Appendix.

\textit{Acknowledgments.-} We acknowledge support from the General Research Fund (GRF) grant No. 17307925 and the Research Impact Fund (RIF) No. R7035-21.

\bibliographystyle{apsrev4-2}
\bibliography{common}

\begin{thebibliography}{99}
\bibitem{PAM+10} S. Pironio, et al. \textit{Random numbers certified by Bell's theorem}. Nature 464, 1021 (2010). 

\bibitem{BRC23} R. Bhavsar, S. Ragy and R. Colbeck. \textit{Improved device-independent randomness expansion rates using two sided randomness}. New J. Phys. 25, 093035 (2023).

\bibitem{LLR+21} W.-Z. Liu et al. \textit{Device-independent randomness expansion against quantum side information}. Nat. Phys. 17, 448 (2021).
\bibitem{MS14} C. A. Miller and Y. Shi. \textit{Robust protocols for securely expanding randomness and distributing keys using untrusted quantum devices}. Proc. of the 46th annual ACM symposium on Theoryof computing (STOC'14), pp. 417-426 (2014). 

\bibitem{FWE+23} C. Foreman, S. Wright, A. Edgington, M. Berta and F. J. Curchod. \textit{Practical randomness amplification and privatisation with implementations on quantum computers}. Quantum 7, 969 (2023). 
\bibitem{KAF20} M. Kessler and R. Arnon-Friedman. \textit{Device-independent Randomness Amplification and Privatization}. IEEE Journal on Selected Areas in Information Theory, vol. 1, no.2, pp. 568-584 (2020).
\bibitem{CR12} R. Colbeck and R. Renner. \textit{Free randomness can be amplified}. Nat. Phys. 8, 450-453 (2012). 
\bibitem{BRG+16} F. G. S. L. Brandao, R. Ramanathan, A. Grudka, K. Horodecki, M. Horodecki, P. Horodecki, T. Szarek and H. Wojewodka. \textit{Realistic noise-tolerant randomness amplification using finite number of devices}. Nat. Comm. 7, 11345 (2016). 
\bibitem{GMT+13} R. Gallego, Ll. Masanes, G. de la Torre, C. Dhara, L. Aolita and A. Ac\'{i}n. \textit{Full randomness from arbitrarily deterministic events}. Nat. Comm. 4, 2654 (2013). 
\bibitem{RBH+16} R. Ramanathan, F. G. S. L. Brandao, K. Horodecki, M. Horodecki, P. Horodecki and H. Wojewodka. \textit{Randomness Amplification under Minimal Fundamental Assumptions on the Devices}. Phys. Rev. Lett. 117, 230501 (2016). 
\bibitem{R23} R. Ramanathan. \textit{Finite Device-Independent Extraction of a Block Min-Entropy Source against Quantum Adversaries}. arXiv:2304.09643 (2023). 

\bibitem{RYS24} R. Ramanathan, Y. Liu and S. Pironio. \textit{When Quantum Nonlocality Does Not Play Dice}. arXiv:2408.03665v2  (2025).

\bibitem{BSS14} J.-D. Bancal, L. Sheridan, and V. Scarani. \textit{More Randomness from the Same Data}. New J. Phys. 16, 033011 (2014). 
\bibitem{NSPS14} O. Nieto-Silleras, S. Pironio and J. Silman. \textit{Using complete measurement statistics for optimal device-independent randomness evaluation}. New J. Phys. 16, 013035 (2014).

\bibitem{NSBSP18} O. Nieto-Silleras, C. Bamps, J. Silman and S. Pironio. \textit{Device-independent randomness generation from several Bell estimators}. New J. Phys. 20, 023049 (2018). 

\bibitem{KRS09} R. Koenig, R. Renner, and C. Schaffner. \textit{The operational meaning of min- and max-entropy}. IEEE Trans. Inf. Th., vol. 55, no. 9 (2009). 

\bibitem{KKMTV08} J. Kempe, H. Kobayashi, K. Matsumoto, B. Toner and T. Vidick. \textit{Entangled games are hard to approximate}. 49th Annual IEEE Symposium on Foundations of Computer Science, FOCS 2008, pp. 447-456 (2008).

\bibitem{DVJI14} J. I. de Vicente. \textit{On nonlocality as a resource theory and nonlocality measures} J. Phys. A: Math. Theor. 47, 424017 (2014). 

\bibitem{REAM12} R. Gallego, L. E. Würflinger, A. Acín and M. Navascués. \textit{An Operational Framework for Nonlocality}. Phys. Rev. Lett. 109, 070401 (2012)

\bibitem{ZLH+25} C. Zhang, Y. Li, X.-M. Hu, Y. Xiang, C.-F. Li, G.-C. Guo, J. Tura, Q. Gong, Q. He and B.-H. Liu. \textit{Randomness versus Nonlocality in Multi-input and Multi-output Quantum Scenario}. Phys. Rev. Lett. 134, 090201 (2025). 

\bibitem{vertesi2010closing}
T. V{\'e}rtesi, S. Pironio, and N. Brunner.
\newblock \textit{Closing the detection loophole in Bell experiments using qudits.}
\newblock {Phys. Rev. Lett.} 104(6):060401, (2010).

\bibitem{WAP21} E. Woodhead, A. Ac\'{i}n and S. Pironio. \textit{Device-independent quantum key distribution with asymmetric CHSH inequalities}. Quantum 5, 443 (2021). 

\bibitem{KW16} J. Kaniewski and S. Wehner. \textit{Device-independent two-party cryptography secure against sequential attacks}. New J. Phys. 18, 055004 (2016). 
\bibitem{WBA18} E. Woodhead, B. Bourdoncle and A. Ac\'{i}n. \textit{Randomness versus nonlocality in the Mermin-Bell experiment with three parties}. Quantum 2, 82 (2018). 

\bibitem{RHAP+18} R. Ramanathan, M. Horodecki, H. Anwer, S. Pironio, K. Horodecki, M. Grunfeld, S. Muhammad, M. Bourennane, P. Horodecki. \textit{Practical No-Signalling proof Randomness Amplification using Hardy paradoxes and its experimental implementation}. arXiv:1810.11648 (2018).

\bibitem{BFF24} P. Brown, H. Fawzi and O. Fawzi. \textit{Device-independent lower bounds on the conditional von Neumann entropy}. Quantum 8, 1445 (2024).



\bibitem{ACP+16} A. Ac\'{i}n, D. Cavalcanti, E. Passaro, S. Pironio and P. Skrzypczyk. \textit{Necessary detection efficiencies for secure quantum key distribution and bound randomness}. Phys. Rev. A 93, 012319 (2016).  

\bibitem{MS17} C. A. Miller and Y. Shi. \textit{Randomness in nonlocal games between mistrustful players}. Quant. Inf. and Comp. 17, No. 7 \& 8, pp. 0595 - 0610 (2017). 
\bibitem{MS17-2} C. A. Miller and Y. Shi. \textit{Certifying the absence of quantum nonlocality}.  Quant. Inf. and Comp. 17, No. 7 \& 8, pp. 0595 - 0610 (2017). 

\bibitem{Win99} A. Winter. \textit{Coding Theorems of Quantum Information Theory}. Ph.D. Disseration, Uni. Bielefeld. arXiv:quant-ph/9907077 (1999).

\bibitem{NPA1} M. Navascues, S. Pironio and A. Ac\'{i}n. \textit{Bounding the set of quantum correlations}. Phys. Rev. Lett. 98, 010401 (2007).
\bibitem{NPA2} M. Navascues, S. Pironio and A. Ac\'{i}n. \textit{A convergent hierarchy of semidefinite programs characterizing the set of quantum correlations}. New J. Phys. 10, 073013 (2008).

\bibitem{FBL+21} M. Farkas, M. Balanz\'{o}-Juand\'{o}, K. {\L}ukanowski, K. Ko{\l}odynski and A. Ac\'{i}n. \textit{Bell nonlocality is not sufficient for the security of standard device-independent quantum key distribution protocols}. Phys. Rev. Lett. 127, 050503 (2021).

\bibitem{mikos2023extremal}
A. Mikos-Nuszkiewicz and J. Kaniewski.
\newblock \textit{Extremal points of the quantum set in the CHSH scenario: conjectured analytical solution.}
\newblock {arXiv:2302.10658} (2023).


\bibitem{SV84} M. Santha and U. V. Vazirani. \textit{Generating quasi-random sequences from slightly-random sources}. Proc. 25th IEEE Symp. Found. Comput. Sci. (FOCS'84) pp. 434-440 (1984).

\bibitem{Hardy93}
L.~Hardy. \textit{ Nonlocality for two particles without inequalities for almost all entangled states.} \newblock {Phys. Rev. Lett.} 71(11):1665, (1993).

\bibitem{Helstrom69} C. W. Helstrom. \textit{Quantum detection and estimation theory}. Journal of Statistical Physics, 1:231–252, (1969).



%%%%%%%%%%%%%%%%%%%%%%%%%%%%%%






% \bibitem{VPB10} T. Vertesi, S. Pironio and N. Brunner. \textit{Closing the detection loophole in Bell experiments using qudits}. Phys. Rev. Lett. 104, 060401 (2010). 



% \bibitem{V17} T. Vidick. \textit{Parallel DIQKD from parallel repetition}. arXiv: 1703.08508 (2017). 


















\end{thebibliography}

% \section{Introduction}

\onecolumngrid
\appendix
\newpage

\section{Device-Independent Guessing Probability} 
\label{sec:DIGP}
The guessing probability $P_g$ quantifies the extent to which an eavesdropper Eve is able to guess the outputs of Alice and Bob. In other words, the randomness of Alice and Bob's device given the side information possessed by Eve is quantified by the min-entropy defined as $H_{\min} = - \log_2 P_g$. Our primary focus in this paper is on the relationship between quantum nonlocality and the device-independent randomness. 

\subsection{Quantum adversaries} 
\subsubsection{Guessing Probability for a fixed setting}
Let us first describe the traditionally defined guessing probability for a quantum adversary (who may hold a device that shares arbitrary quantum correlations with Alice and Bob's device). In this case, the joint state of the devices is an arbitrary quantum state $\rho_{\mathcal{A}\mathcal{B}\mathcal{E}}$. After Alice and Bob have performed specific measurements $x^*, y^*$ and obtained classical outcomes $a, b$ with probabilities $p_{AB|XY}(a,b|x^*,y^*)$, the correlations are described by the classical-quantum state $\sum_{a,b} p_{AB|XY}(a,b|x^*, y^*)|a,b \rangle \langle a, b| \otimes \rho_{\mathcal{E}}^{a,b,x^*,y^*}$, where $\rho_{\mathcal{E}}^{a,b,x^*,y^*}$ is the state of Eve's device given that Alice and Bob input $x^*, y^*$ and obtained outputs $a, b$ respectively. Eve now performs a measurement described by a $d^2$-element POVM $\{O_{e|z}\}$ in an attempt to guess Alice and Bob's outcome $(a,b)$. Eve outputs $e$ as her best guess of $(a,b)$ which happens with probability $p_{E|ZABXY}(e = (a,b)|z, a,b,x^*, y^*) = \text{Tr}\left(O_{e|z} \rho_{\mathcal{E}}^{a,b,x^*, y^*} \right)$. We obtain that Eve's guessing probability for the fixed setting $x^*, y^*$ is given by
\begin{eqnarray}
    &&\max_{\{O_{e|z}\}} \sum_{a,b} p_{AB|XY}(a,b|x^*, y^*)p_{E|ZABXY}(e=(a,b)|z,a,b,x^*, y^*) \nonumber \\
    &=& \max_{\{O_{e|z}\}} \sum_{a,b} p_{E|Z}(e=(a,b)|z)p_{AB|XYEZ}(a,b|x^*,y^*,e=(a,b),z),
\end{eqnarray}
%where $p^{Q}$ denotes the probabilities obtained using the quantum realization $Q$ and 
where we have used Bayesian rewriting. In other words, Eve performs a convex decomposition of Alice-Bob's behavior $\textbf{p} = \{ p_{AB|XY}\}$ into conditional behaviors $\textbf{p}^{a,b} = \{p_{AB|XYEZ}(a,b|x^*,y^*,e=(a,b),z) \}$ where $p_{E|Z}(e=(a,b)|z)$ define the weights in the convex decomposition.

We arrive at the usual definition of the device-independent guessing probability $P^{(q)}_g$ for a quantum adversary for a fixed setting pair $(x^*, y^*)$ when Alice and Bob observe a behavior $\textbf{p}$ as
\begin{eqnarray}
\label{eq:gpi}
   P^{(q)}_g(\textbf{p}, (x^*, y^*)) &=& \max_{\{\textbf{p}^{a,b}\}, \{q_{a,b}\}} \sum_{a,b} q_{a,b} {p}^{a,b}_{AB|XY}(a,b|x^*,y^*) \nonumber \\
    &&\text{s.t.} \;\;\; \sum_{a,b} q_{a,b} {p}^{a,b}_{AB|XY}(a',b'|x',y') = {p}_{AB|XY}(a',b'|x',y') \; \; \; \; \forall a',b',x',y' \nonumber \\
    && \;\;\qquad \textbf{p}^{a,b} \in \mathcal{Q}^{bip}, \; \; \; q_{a,b} \geq 0 \; \; \forall a, b, \;\; \;  \sum_{a,b} q_{a,b} = 1.
\end{eqnarray}
It is traditional to consider a semidefinite programming relaxation of the convex set $\mathcal{Q}^{bip}$ using the well-known NPA hierarchy \cite{NPA1, NPA2}. We also note that the constraint on reproducing the observed behavior $\textbf{p}$ is in some cases replaced by the constraint of an observed value $\omega_{obs}$ in some nonlocal game characterized by an input distribution $\nu_{XY}(x',y')$ and a winning predicate $V(a',b',x',y') \in \{0,1\}$ (an observed violation of a chosen Bell inequality). That is, we replace the constraint of reproducing the observed behavior $\textbf{p}$ in \eqref{eq:gpi} by the constraint
\begin{equation}
    \sum_{x',y',a',b'} \nu_{XY}(x',y') V(a',b',x',y') \sum_{a,b} q_{a,b} {p}^{a,b}_{AB|XY}(a',b'|x',y') = \omega_{obs},
\end{equation}
to define the variant of the guessing probability denoted $P_g(\omega_{obs}, (x^*,y^*))$. 

As a further simplification, the guessing probability expression is sometimes written in terms of a family of non-normalised behaviors $\widetilde{\textbf{p}}^{a,b,x^*,y^*}$ defined as
\begin{equation}
    \widetilde{\textbf{p}}^{a,b,x^*,y^*} = p_{E|Z}(e=(a,b)|z) \textbf{p}^{a,b,z},
\end{equation}
where 
\begin{equation}
    p^{a,b,z}(a',b'|x,y) = p_{AB|XYEZ}(a',b'|x,y,e=(a,b),z)
\end{equation}
are conditional behaviors conditioned on the outcomes $e=(a,b)$ of Eve for her measurement $z$. We note by the property of the convex decomposition that $\sum_{a,b} \widetilde{p}^{a,b,}(a',b'|x',y') = p(a',b'|x',y')$ for all $a',b' \in \{1,\ldots, d\}$ and $x', y' \in \{1,\ldots, m\}$. Crucially, every set of non-normalised behaviors $\widetilde{\textbf{p}}^{a,b}$ belonging to the corresponding non-normalised bipartite quantum set $\widetilde{\mathcal{Q}}^{bip}$ satisfying $\sum_{a,b} \widetilde{\textbf{p}}^{a,b} = \textbf{p}$ can be interpreted as arising from some quantum realisation and a POVM measurement $\{O_{e|z}\}$ performed by Eve. Here $\widetilde{\mathcal{Q}}^{bip}$ denotes the set of non-normalised quantum behaviors, i.e., behaviors obtainable from measurements on non-normalised quantum states $\widetilde{\rho}_{\mathcal{A}\mathcal{B}}$ (i.e. with $\text{Tr}\left(\widetilde{\rho}_{\mathcal{A}\mathcal{B}} \right) \geq 0$). The device-independent guessing probability for a quantum adversary for a fixed setting pair $(x^*, y^*)$ when Alice and Bob observe a behavior $\textbf{p}$ then simplifies to
\begin{eqnarray}
\label{eq:gpi}
   P^{(q)}_g(\textbf{p}, (x^*, y^*)) &=& \max_{\widetilde{\textbf{p}}^{a,b}} \sum_{a,b} \widetilde{p}^{a,b}_{AB|XY}(a,b|x^*,y^*) \nonumber \\
    &&\text{s.t.} \;\;\; \sum_{a,b} \widetilde{p}^{a,b}_{AB|XY}(a',b'|x',y') = p_{AB|XY}(a',b'|x',y') \; \; \; \; \forall a',b',x',y' \nonumber \\
    && \;\;\qquad \widetilde{\textbf{p}}^{a,b} \in \widetilde{\mathcal{Q}}^{bip}.
\end{eqnarray}
%It is traditional to consider a semidefinite programming relaxation of the convex set $\widetilde{\mathcal{Q}}^{bip}$ using the well-known NPA hierarchy \cite{NPA1, NPA2}. We also note that 
Again, the constraint can be similarly modified to reflect an observed value $\omega_{obs}$ in some nonlocal game. 
%Again, the constraint on reproducing the observed behavior $\textbf{p}$ can be replaced by the constraint of an observed value $\omega_{obs}$ in some nonlocal game characterized by an input distribution $\nu(x',y')$ and a winning predicate $V(a',b',x',y') \in \{0,1\}$ (an observed violation of a chosen Bell inequality). That is, we replace the constraint of reproducing the observed behavior $\textbf{p}$ by the constraint
%\begin{equation}
%    \sum_{x',y',a',b'} \nu(x',y') V(a',b',x',y') \sum_{a,b}\widetilde{\textbf{p}}^{a,b}_{A,B|X,Y}(a',b'|x',y') = \omega_{obs},
%\end{equation}
%to define the variant of the guessing probability denoted $P_g(\omega_{obs}, (x^*,y^*))$. 

\subsubsection{Average Guessing Probability over all settings}
We now proceed to consider Eve's guessing probability of the measurement outcomes of Alice and Bob's device on average over all their measurement settings $(x,y)$ given that Eve has some classical side information regarding the input distribution $\mu_{XY}(x,y)$. Specifically, here we consider the scenario where even in the randomness generation rounds, Alice and Bob choose their inputs randomly with probabilities given by $\mu_{XY}$. We denote the corresponding guessing probability as $P_g^{(q)}(\textbf{p}, \mu_{XY})$ with a corresponding average min-entropy given by $- \log_2 \left(P_g^{(q)}(\textbf{p}, \mu_{XY}) \right)$.

As before, for each measurement setting $(x,y)$, the correlations are described by a classical-quantum state  $\sum_{a,b} p_{AB|XY}(a,b|x, y)|a,b \rangle \langle a, b| \otimes \rho_{\mathcal{E}}^{a,b,x,y}$ and Eve performs a measurement $\{O_{e|z}\}$ after coming to know the inputs $(x,y)$. In other words, Eve chooses an optimal measurement $z = (x,y)$ for each input $(x,y)$ and therefore performs $|X||Y| = m^2$ convex decompositions of Alice-Bob's behavior $\textbf{p}$. 

In this case, we obtain the average device-independent guessing probability for a quantum adversary over all settings chosen with probabilities $\mu_{XY}$ when Alice and Bob observe a behavior $\textbf{p}$ as
\begin{eqnarray}
\label{eq:gpii}
 P^{(q)}_{g}(\textbf{p}, \mu_{XY}) &=& \max_{\{p_{ABE|XYZ}\}} \sum_{x,y} \mu_{XY}(x,y) \sum_{a,b} p_{ABE|XYZ}(a,b,e=(a,b)|x,y,z=(x,y)) \nonumber \\ 
   && \; \; \; \;\;\; \text{s.t.} \quad \sum_{e'} p_{ABE|XYZ}(a',b',e'|x',y',z') = {p}_{AB|XY}(a',b'|x',y') \; \; \; \; \forall a',b',x',y',z' \nonumber \\
   && \qquad \qquad \;  \{p_{ABE|XYZ} \}\in \mathcal{Q}^{trip},
\end{eqnarray}
Observe that in this case, it is necessary to consider the set of tripartite quantum correlations $\mathcal{Q}^{trip}$ (where Eve has $|X||Y|$ inputs and $|A||B|$ outputs per input) which again can be relaxed to a semidefinite programming outer approximation using the NPA hierarchy. Also similarly to the case of a fixed setting, we might instead consider the constraint of an observed value $\omega_{obs}$ in some nonlocal game characterized by an input distribution $\nu_{XY}(x',y')$ and a winning predicate $V(a',b',x',y') \in \{0,1\}$ 
\begin{equation}
    \sum_{x',y',a',b',e'} \nu_{XY}(x',y') V(a',b',x',y') p_{ABE|XYZ}(a',b',e'|x',y',z') = \omega_{obs} \; \; \; \; \forall z'
\end{equation}
to define the variant of the guessing probability denoted $P_g(\omega_{obs}, \mu_{XY})$. 

%\textit{Summary of DIGP for quantum adversaries.-} 
To summarize, for a quantum adversary, in the case of a fixed setting pair $(x^*, y^*)$, we consider a single convex decomposition of the observed behavior $\textbf{p}$, while in the case of an average over all settings, Eve performs multiple ($|X||Y|$) convex decompositions of the observed behavior. We now note a crucial difference between these two scenarios for the quantum adversary which arises as a consequence of the results in \cite{RYS24} and Thm. \ref{thm:no-bound-rand} in Appendix \ref{app:no-bound-rand}. Namely, that there exist maximally nonlocal quantum behaviors that can lead to zero randomness in the fixed-setting scenario for any pair of settings $(x^*, y^*)$, and yet certify non-zero randomness in the average scenario where measurement settings are chosen randomly during the generation rounds. 

%In the latter case, since all convex decompositions are required to arise from a single tripartite quantum behavior, in general it is the case that

\begin{lem}
There exist maximally nonlocal quantum behaviors $\textbf{p} \in \mathcal{Q}^{bip}$ for which $P_g^{(q)}(\textbf{p}, (x^*, y^*)) = 1$ for any fixed $(x^*, y^*) \in |X| \times |Y|$, while at the same time  $P_g^{(q)}(\textbf{p}, \mu_{XY}) < 1$ for any input distribution such that $\mu_{XY}(x,y) > 0$ for all $x, y$. 
%$\mu(\cdot)$$ \neq \delta_{x^*,y^*}$ for some $(x^*,y^*)$. 
\label{lem:1}
\end{lem}
\begin{proof}
In \cite{RYS24}, some of us identified an extreme case of a nonlocal behavior $\textbf{p}$ that exhibits maximum nonlocality in the sense of achieving maximum quantum value for a Bell expression, and yet has $P_g^{(q)}(\textbf{p}, (x^*, y^*)) = 1$ for any fixed input pair $(x^*, y^*)$. In contrast for the same quantum behavior $\textbf{p}$, following the proof of Thm. \ref{thm:no-bound-rand} we see that $P_g^{(q)}(\textbf{p}, \mu_{XY}) < 1$ for any input distribution such that $\mu_{XY}(x,y) > 0$ for all $x, y$.
\end{proof}

\subsection{Classical adversaries}
We now proceed to consider the case of classical adversaries, i.e., adversaries who can prepare arbitrary quantum devices for Alice and Bob but only share classical correlations with their devices (we say adversary's adversary's side information about Alice and Bob's devices is classical). In this case, the joint state of Alice, Bob and Eve's devices is described by a classical-quantum state $\sum_{\lambda} q_{\lambda} \rho^{(\lambda)}_{\mathcal{A}\mathcal{B}} \otimes |\lambda \rangle \langle \lambda|$. As in the case of a quantum adversary, Eve performs measurements described by a $d^2$-element POVM $\{O_{e|z}\}$ in an attempt to guess Alice and Bob's outcome pair $(a,b)$. 
In this case, the optimization is performed only over behaviors $p_{ABE|XYZ}$ of the form $\sum_{\lambda} q_{\lambda} p^{\lambda}_{AB|XY} \otimes p^{\lambda, \text{det}}_{E|Z} \in \mathcal{Q}^{trip}$, where $p^{\lambda, \text{det}}_{E|Z}$ denotes a behavior that returns deterministic outcomes for any $z$. In other words, in the case of a classical adversary, the behavior $p_{ABE|XYZ}$ is local across the cut $AB|E$, so that we only optimize over such behaviors instead of the entire set of tripartite quantum behaviors $\mathcal{Q}^{trip}$. 

The device-independent guessing probability for a classical adversary for a fixed setting pair $(x^*, y^*)$ is still given by the Guessing Probability expression in \eqref{eq:gpi}. That is, we consider a convex decomposition of the given behavior $\textbf{p}$ into bipartite quantum behaviors
\begin{eqnarray}
    p_{AB|XY}(a',b'|x',y') = \sum_{\lambda} q_{\lambda} p^{\lambda}_{AB|XY}(a',b'|x',y') \; \; \; \forall a',b',x',y',
\end{eqnarray}
that is optimized for the setting $(x^*, y^*)$ and the guessing probability expression is still given by
\eqref{eq:gpi}, but now with $\lambda$ taking the role of $a,b$ and with
\begin{eqnarray}
    \widetilde{p}^{a,b}_{AB|XY}(a',b'|x',y') = q_{a,b} p^{a,b}_{AB|XY}(a',b'|x',y'). 
\end{eqnarray}

The situation changes dramatically when we consider the scenario where Alice and Bob choose their inputs randomly with probabilities given by $\mu_{XY}$ during the randomness generation rounds. In this case, Eve is only able to perform a \textit{single} convex decomposition of $\textbf{p}$ with weights $q_{\lambda}$ with this decomposition now being optimized to guess the outcomes for all settings. We arrive at the guessing probability expression 
\begin{eqnarray}
    \label{eq:gpiii}
 P^{(c)}_{g}(\textbf{p}, \mu_{XY}) &=& \max_{\widetilde{\textbf{p}}^{a,b}} \sum_{x,y} \mu_{X,Y}(x,y) \sum_{a,b} \widetilde{p}^{a,b}(a_{x,y},b_{x,y}|x,y) \nonumber \\ 
   &&\text{s.t.} \;\;\; \sum_{a,b} \widetilde{p}^{a,b}_{AB|XY}(a',b'|x',y') = p_{AB|XY}(a',b'|x',y') \; \; \; \; \forall a',b',x',y' \nonumber \\
    && \;\;\qquad \widetilde{\textbf{p}}^{a,b} \in \widetilde{\mathcal{Q}}^{bip},
\end{eqnarray}
where now $a_{x,y}, b_{x,y} \in \{1,\ldots, d\}$ are specific outcomes chosen for each $(x,y)$ to maximize the guessing probability. This is the form of the guessing probability advocated in \cite{BSS14, NSPS14} to obtain more randomness from the same data, i.e., to use the entire dataset rather than the value of a single Bell expression and to obtain randomness from all settings rather than a fixed $(x^*, y^*)$. As in the case of a quantum adversary, we may also modify the optimization program to reflect the constraint of an observed value $\omega_{obs}$ in a suitably chosen nonlocal game. 

%\textit{Summary of DIGP for classical adversaries.-} 
To summarize, for a classical adversary, in both cases of a fixed setting pair $(x^*, y^*)$ and of an average over all setting pairs, Eve performs a single convex decomposition of the observed behavior $\textbf{p}$. The case of a fixed setting pair results in the same guessing probability as for a quantum adversary, whereas the case of an average guessing probability shows a difference between classical and quantum adversaries. Furthermore, since the convex decomposition optimized for a fixed setting pair may be different than the one optimized for an average over setting pairs, in general it is again the case that
\begin{eqnarray}
    P_g^{(c)}(\textbf{p}, \mu_{XY}) \neq \sum_{x^*, y^*} \mu_{XY}(x^*, y^*) P_g^{(c)}(\textbf{p}, (x^*, y^*)).
\end{eqnarray}

\subsection{No-Signalling Adversaries}
In the case of an adversary who is only restricted by the no-signalling principle, the optimization sets are replaced by the corresponding no-signalling versions denoted by $\mathcal{NS}$. Again, consider the scenario where Alice and Bob choose their inputs randomly with probabilities given by $\mu_{XY}$ during the randomness generation rounds. In this case, we obtain  the average device-independent guessing probability for a no-signalling adversary $P^{(ns)}_{g}$ Eve over all settings chosen with probabilities $\mu_{XY}$ when Alice and Bob observe a behavior $\textbf{p}$ as 
\begin{eqnarray}
\label{eq:gpiv}
    P^{(ns)}_{g}(\textbf{p}, \mu_{XY}) &=& \max_{\{p_{ABE|XYZ}\}} \sum_{x,y} \mu_{XY}(x,y) \sum_{a,b} p_{ABE|XYZ}(a,b,e=(a,b)|x,y,z=(x,y)) \nonumber \\ 
   && \; \; \; \;\;\; \text{s.t.} \quad \sum_{e'} p_{ABE|XYZ}(a',b',e'|x',y',z') = p_{AB|XY}(a',b'|x',y') \; \; \; \; \forall a',b',x',y',z' \nonumber \\
   && \qquad \qquad \;  \{ p_{ABE|XYZ} \} \in \mathcal{NS}^{trip}.
\end{eqnarray}
And correspondingly, the guessing probability optimized for a fixed setting pair $(x^*, y^*)$ is given as
\begin{eqnarray}
    P^{(ns)}_{g}(\textbf{p}, (x^*, y^*)) &=& \max_{\{p_{ABE|XYZ}\}} \sum_{a,b} p_{ABE|XYZ}(a,b,e=(a,b)|x,y,z=(x^*,y^*)) \nonumber \\ 
   && \; \; \; \;\;\; \text{s.t.} \quad \sum_{e'} p_{ABE|XYZ}(a',b',e'|x',y',z=(x^*,y^*)) = \textbf{p}_{AB|XY}(a',b'|x',y') \; \; \; \; \forall a',b',x',y', \nonumber \\
   && \qquad \qquad \; \{ p_{ABE|XYZ} \} \in \mathcal{NS}^{trip}.
\end{eqnarray}
which can be rewritten as
\begin{eqnarray}
      P^{(ns)}_g(\textbf{p}, (x^*, y^*)) &=& \max_{\widetilde{\textbf{p}}^{a,b}} \sum_{a,b} \widetilde{p}^{a,b}_{AB|XY}(a,b|x^*,y^*) \nonumber \\
    &&\text{s.t.} \;\;\; \sum_{a,b} \widetilde{p}^{a,b}_{AB|XY}(a',b'|x',y') = p_{AB|XY}(a',b'|x',y') \; \; \; \; \forall a',b',x',y' \nonumber \\
    && \;\;\qquad \widetilde{\textbf{p}}^{a,b} \in \widetilde{\mathcal{NS}}^{bip}.
\end{eqnarray}
where $\widetilde{\mathcal{NS}}^{bip}$ denotes the set of non-normalized no-signalling behaviors and we have chosen 
\begin{eqnarray}
      \widetilde{p}^{a,b}_{AB|XY}(a',b'|x',y') = p_{E|Z}(e=(a,b)|z=(z^*,y^*)) p_{AB|XYEZ}(a',b'|x',y',e=(a,b),z=(x^*,y^*)). 
\end{eqnarray}
In other words, for a fixed measurement setting, we optimize over all convex decompositions of the observed behavior $\textbf{p}$ into no-signalling behaviors. 

%\textit{Summary of DIGP for no-signalling adversaries.-} 
To summarize, for a no-signalling adversary, in the case of a fixed setting pair $(x^*, y^*)$, we consider a single convex decomposition of the observed behavior $\textbf{p}$ but this time into no-signalling behaviors, while in the case of an average over all settings Eve performs multiple convex decompositions into no-signalling behaviors. Again, all convex decompositions are required to arise from a single tripartite no-signalling behavior. However, in this case in \cite{ACP+16} it was proven that the average guessing probability is precisely equal to the average over all settings of the fixed-setting guessing probability. Combined with the fact that there exist maximally nonlocal quantum behaviors that certify zero randomness against NS adversaries in the fixed-setting scenario for any pair of settings $(x^*, y^*)$, proves the existence of \textit{bound randomness} against the no-signalling adversary.

\begin{lem}[\cite{ACP+16}]
Let $\textbf{p}$ denote a bipartite no-signaling behavior and let $\mu_{XY}$ denote a probability distribution for the inputs. Then it holds that
\begin{eqnarray}
     P^{(ns)}_{g}(\textbf{p}, \mu_{XY}) = \sum_{x^*,y^*} \mu_{XY}(x^*,y^*) P^{(ns)}_g(\textbf{p}, (x^*, y^*)).
\end{eqnarray}
\label{lem:2}
\end{lem}

%Since we are now maximizing the linear objective function in \eqref{eq:gpii} with linear constraints over the convex set of such behaviors, the maximum is achieved at an extremal point, i.e., one of the form $p_{AB|XY} \otimes p^{\text{det}}_{E|Z}$. The optimization is now reduced to 

%Let us write the objective function in \eqref{eq:gpii} as
%\begin{eqnarray}
%    &&\sum_{x,y} p(x,y) \sum_{a,b} p_{A,B,E|X,Y,Z}(a,b,e=(a,b)|x,y,z=(x,y)) \nonumber \\ 
%    &=& \sum_{x,y} p(x,y) \sum_{e} p_{E|Z}(e|z=(x,y)) p_{AB|XYEZ}(e|x,y,e,z=(x,y)).
%\end{eqnarray}
%Since for each $(x,y)$, 

\section{No Bound Randomness against Quantum Adversaries}
\label{app:no-bound-rand}
%%%%%%%%%%%%%%%%%%%%%%%%%%%%%%%%%%%%%%%%%%%%%%%%%%%%%%%%%%%%%%%%%%%
In the previous appendix, we have seen the existence of bound randomness against no-signalling adversaries, while in Lemma \ref{lem:1} we saw that certain candidate maximally nonlocal quantum behaviors that certify zero randomness against quantum adversaries in the fixed-setting scenario, still do not exhibit bound randomness when the settings are chosen with a full-support distribution $\mu_{XY}$. In this section, we prove that in fact, all nonlocal quantum behaviors certify randomness in this average scenario, i.e., there does not exist any bound randomness against quantum adversaries.

Recall that the value of a general two-player nonlocal game $G$ (aka two-party Bell expression) is written as
\begin{equation}
\omega(G) := \sum_{x,y,a,b} \nu_{XY}(x,y) V(a,b,x,y) P_{AB|XY}(a,b|x,y),
\end{equation}
where $x \in \mathsf{X}, y \in\mathsf{Y}$ denote the inputs of the two players Alice and Bob respectively, while $a \in \mathsf{A}, b \in \mathsf{B}$ denote their respective outputs. The probability distribution with which the inputs are chosen 
%(by a referee in a nonlocal game) 
is denoted by $\nu_{XY}$ with $0 \leq \nu_{XY}(x,y) \leq 1$ for any $x,y$ - when the distribution is of product form $\nu^{(A)}_{X} \cdot \nu^{(B)}_{Y}$, the game is said to be free.
%The predicate $V(a,b,x,y) \in \{0,1\}$ defines the winning condition for the game, and the set of conditional distributions $\{P_{AB|XY}(a,b|x,y)\}$ is considered to belong to some well-defined theory, such as classical $\mathcal{C}$, quantum $\mathcal{Q}$, and no-signalling $\mathcal{NS}$ theories. In general, $\mathcal{C} \subseteq \mathcal{Q} \subseteq \mathcal{NS}$. 

\begin{defi}
We say that a nonlocal quantum behavior $P = \{P_{AB|XY}(a,b|x,y)\} \in \mathcal{Q}$ is quantum predictable or exhibits bound randomness against a quantum adversary, if there is a strategy involving quantum side information that allows a third player Eve, when given the pair of inputs $(x,y)$ chosen by Alice and Bob, to perfectly guess their outputs $(a,b)$, for any $(x,y) \in \mathsf{X} \times \mathsf{Y}$. Analogously, if Eve is able to predict Alice's (Bob's) output $a$ ($b$) when given their chosen input $x$ ($y$) for any $x \in \mathsf{X}$ ($y \in \mathsf{Y}$), we say that $P$ exhibits bound local randomness against a quantum adversaries.
\end{defi}

In this Appendix, we aim to show that no nonlocal quantum behavior exhibits bound randomness against a quantum adversary \cite{MS17, MS17-2}. Our proof will use the immunization technique from \cite{KKMTV08} where the method was used to show the NP-hardness of deciding whether the quantum value of three-player games is $1$ or $\leq 1 - \epsilon$ for small $\epsilon > 0$. We extend their approach which was restricted to Pseudo-telepathy games (games in which the quantum value is unity, $\omega_{\mathcal{Q}} = 1$) to general nonlocal games to show that for any nonlocal game the observation of a super-classical winning probability by the honest parties implies that a quantum adversary cannot guess the outcomes perfectly. 

We remark that while the considerations in this Appendix will be restricted to bipartite behaviors for ease of notation, it will be apparent from the proof that the statement holds for general $n$-partite nonlocal quantum behaviors for any $n \geq 2$. Also note that the scenario considered in the definition is the usual scenario of the guessing probability in general public-source randomness amplification protocols as well as in certain randomness expansion protocols where multiple input pairs $(x,y)$ are used in the randomness generation rounds. As such, we aim to show that all  nonlocal games serve to certify randomness for such protocols. 
%i.e., the observation of a super-classical score for any nonlocal game implies that the outputs of the players cannot be perfectly guessed by a quantum adversary. 
Finally, we note that this is in contrast to the situation of a no-signalling adversary where it is well known that there exist even extremal nonlocal no-signalling behaviors $P \in \mathcal{NS}$ that exhibit bound randomness, i.e., there is a strategy involving no-signalling side information that allows Eve to guess the outputs $(a,b)$ when given any pair of inputs $(x,y)$, even as Alice and Bob observe the maximum value $\omega_{\mathcal{NS}}(G)$ .

\begin{thm}
\label{thm:no-bound-rand}
No nonlocal quantum behavior exhibits bound randomness against a quantum adversary. 
\end{thm}
\begin{proof}
%Clearly, being able to guess a single player's outcome is a prerequisite to guessing both players' outcomes, so that we can conclude that a behavior does not exhibit bound randomness when it does not exhibit bound local randomness. As such, 
We will prove the statement by contradiction, namely that if Eve is able to guess Alice and Bob's outcomes perfectly, then the behavior is local. Consider a two-player nonlocal game $G$ with inputs $x \in \mathsf{X}, y \in \mathsf{Y}$ and outputs $a \in \mathsf{A}, b \in \mathsf{B}$ respectively, input distribution $\nu_{XY}$ and winning condition given by predicate $V(a,b,x,y)$. Suppose that Alice and Bob observe a value $\omega_{obs}(G) > \omega_{\mathcal{C}}(G)$ and at the same time Eve is able to perfectly guess Alice and Bob's outputs $a, b$ upon receiving their chosen inputs $x, y$ for any $x \in \mathsf{X}, y \in \mathsf{Y}$. That is, suppose there exists a tripartite quantum behaviour $\left\{P_{ABE|XYZ}(a,b,e|x,y,z) \right\}$ such that
\begin{eqnarray}
\label{eq:guessing-game}
\sum_{x,y,a,b} \nu_{XY}(x,y) V(a,b,x,y)  P_{AB|XY}(a,b|x,y) &=& \omega_{obs}(G), \nonumber \\
\sum_{ab} P_{ABE|XYZ}\left(a,b,e=(a,b) \big|x,y,z=(x,y) \right) &=& 1 \; \; \; \; \forall (x,y) \in \mathsf{X} \times \mathsf{Y}.
\end{eqnarray}
Here, $\{P_{AB|XY}(a,b|x,y)\}$ is the marginal behavior of the tripartite quantum behavior $\left\{P_{ABE|XYZ}(a,b,e|x,y,z) \right\}$ which corresponds to a quantum strategy involving some tripartite quantum state $\rho^{ABC}$ and sets of measurements $\{\Pi_{x}^{a}\}_x$, $\{ \Lambda_{y}^b\}_y$ and $\{ \Gamma_{z}^{c} \}_z$ for the three players respectively where the $\Pi_{x}^{a}$ (similarly $\Lambda_{y}^b$ and $\Gamma_{z}^{c}$) are POVM elements, i.e., $\Pi_{x}^{a} \geq 0 \; \forall a,x$ and $\sum_{a} \Pi_{x}^{a} = \mathbf{1} \; \forall x$. 

%We first define a corresponding guessing game $G_{guess}$ for three players Alice, Bob and Eve as follows. In the game $G_{guess}$, the three players have inputs $x \in \mathsf{X}, y \in \mathsf{Y}, z \in \mathsf{Z}$ respectively with $\mathsf{Z} = \mathsf{X} \times \mathsf{Y}$ and outputs $a \in \mathsf{A}, b \in \mathsf{B}, e \in \mathsf{E}$ with $\mathsf{E} = \mathsf{A} \times \mathsf{B}$. The input distribution $\mu_{XYZ}$ is given as $\mu_{XYZ}(x,y,z) := \delta_{(x,y),z} \mu_{XY}(x,y)$ with $\delta_{(x,y),z}$ denoting the Kronecker delta, i.e., $\delta_{(x,y),z} = 1$ for $z = (x,y)$ and $0$ otherwise. The winning condition of $G_{guess}$ is given by the predicate $V_{guess}(a,b,e,x,y,z)$ defined as $V_{guess}(a,b,e,x,y,z) := V(a,b,x,y) W(a,b,e,x,y,z)$ where
%\begin{equation}
%W(a,b,e,x,y,z) = \left\{ 
%\begin{array}{ll}
%1 & \; \text{if} \; e = (a,b) \wedge z=(x,y) \\
%0 & \text{otherwise}.
%\end{array}
%\right.
%\end{equation}
%Clearly, the value achieved by the quantum behavior $\left\{P_{ABE|XYZ}(a,b,e|x,y,z) \right\}$ in this game is also  $\omega(G)$.

Now the second condition in \eqref{eq:guessing-game} says that
\begin{eqnarray}
\label{eq:state-fixed}
\sum_{a,b} \tr \left(\Pi_{x}^{a} \otimes \Lambda_{y}^{b} \otimes \Gamma_{(x,y)}^{(a,b)}  \rho^{ABC} \right) = 1 \quad \forall (x,y) \in \mathsf{X} \times \mathsf{Y}.
% \nonumber \\
%\implies \sum_{a} \langle \Psi' | \Lambda_{(x,1)}^{a} \otimes \Lambda_{(x,3)}^{a} \otimes \mathbf{1} | \Psi' \rangle = 1, \nonumber \\
\end{eqnarray}
In satisfying this condition, we see that Alice and Bob's marginal state $\rho^{AB} = \tr_{E} |\Psi\rangle \langle \Psi|$ (with $E$ denoting the Hilbert space of Eve's system) is fixed by Eve's measurement $\Gamma_{z}$ for any $z = (x,y)$. That is, if Eve observes  outcome $(a,b)$ for her measurement $\{\Gamma_{(x,y)}\}$, then Alice and Bob also definitely observe outcomes $a$ and $b$ for their measurements $\{ \Pi_{x} \}$ and $\{\Lambda_{y} \}$ respectively. In other words, Alice and Bob's measurements do not disturb their marginal state $\rho^{AB}$. To elaborate, let $\Pi_{x} \otimes \Lambda_{y}$ denote the superoperator corresponding to Alice and Bob's measurement $(x,y)$, so that $\Pi_{x} \otimes \Lambda_y \left(\rho^{AB} \right) = \sum_{a,b} \left(\sqrt{\Pi_{x}^a} \otimes \sqrt{\Lambda_{y}^b} \right) \rho^{AB} \left( \sqrt{\Pi_{x}^a} \otimes \sqrt{\Lambda_{y}^b} \right)$ denotes the post-measurement state after performing $\{\Pi_{x}^{a} \otimes \Lambda_{y}^b\}$ on $\rho^{AB}$.

\begin{prop}
\label{prop:fixed-state}
It holds that
\begin{equation}
\Pi_{x} \otimes \Lambda_y \left(\rho^{AB} \right) = \rho^{AB}.
\end{equation}
\end{prop}
\begin{proof}
This can be seen from
\begin{equation}
\big\| \Pi_x \otimes \Lambda_y \left(\rho^{AB} \right) - \rho^{AB} \big\|_1 = 0,
\end{equation}
where $\| \cdot \|_1$ denotes the trace norm, which serves to prove the statement. Since $\Pi_x \otimes \Lambda_y \left(\rho^{AB} \right) = \tr_{E} \Pi_x \otimes \Lambda_y \otimes \mathbf{1} \left(\rho^{ABC} \right)$ and $\rho^{AB} = \tr_{E} \mathbf{1} \otimes \mathbf{1} \otimes \Gamma_{(x,y)} \left(\rho^{ABC} \right)$ we have that by the monotonicity of the trace distance under partial trace,
\begin{equation}
\big\| \Pi_x \otimes \Lambda_y \left(\rho^{AB} \right) - \rho^{AB} \big\|_1 \leq \big\| \Pi_x \otimes \Lambda_y \otimes \mathbf{1} \left(\rho^{ABC} \right) - \mathbf{1} \otimes \mathbf{1} \otimes \Gamma_{(x,y)} \left( \rho^{ABC} \right) \big\|_1.
\end{equation}
Using the triangle inequality we then obtain that
\begin{eqnarray}
\label{eq:tracenorm-bound}
\big\| \Pi_x \otimes \Lambda_y \left(\rho^{AB} \right) - \rho^{AB} \big\|_1 &\leq& \big\| \Pi_x \otimes \Lambda_y \otimes \mathbf{1} \left(\rho^{ABC} \right) - \Pi_x \otimes \Lambda_y \otimes \Gamma_{(x,y)} \left( \rho^{ABC} \right) \big\|_1  \nonumber \\
& +& \big\| \Pi_x \otimes \Lambda_y \otimes \Gamma_{(x,y)} \left( \rho^{ABC} \right)  - \mathbf{1} \otimes \mathbf{1} \otimes \Gamma_{(x,y)} \left( \rho^{ABC} \right) \big\|_1.
\end{eqnarray}
Now, by the gentle measurement lemma \cite{Win99}, we know that for any state $\sigma$ and any measurement operator $0 \leq A \leq \mathbf{1}$, it holds that 
\begin{equation}
\big\| \sigma - \sqrt{A} \sigma \sqrt{A} \|_1 \leq  3 \sqrt{1 - \tr(A \sigma)}.
\end{equation} 
We can bound each of the terms on the right in \eqref{eq:tracenorm-bound} using the above inequality as follows. For the first term, take $\sigma = \oplus_{a,b} \left(\sqrt{\Pi_{x}^{a}} \otimes \sqrt{\Lambda_{y}^b} \otimes \mathbf{1} \right) \rho^{ABC} \left(\sqrt{\Pi_{x}^{a}} \otimes \sqrt{\Lambda_{y}^b} \otimes \mathbf{1} \right)$ and $A = \oplus_{a,b} \mathbf{1} \otimes \mathbf{1} \otimes \Gamma_{(x,y)}^{(a,b)}$, we have that $\tr(A \sigma) = 1$ so that the term is identically $0
$. Similarly, for the second term, take $\sigma = \oplus_{a,b} \left(\mathbf{1} \otimes \mathbf{1} \otimes \sqrt{\Gamma_{(x,y)}^{(a,b)}} \right) \rho^{ABC} \left(\mathbf{1} \otimes \mathbf{1} \otimes \sqrt{\Gamma_{(x,y)}^{(a,b)}} \right)$ and $A = \oplus_{a,b} \Pi_{x}^a \otimes \Lambda_{y}^b \otimes \mathbf{1}$, we have that $\tr(A \sigma) = 1$ so that the term is also identically $0$.

Putting things together we obtain that 
\begin{equation}
\big\| \Pi_{x} \otimes \Lambda_{y} \left(\rho^{AB} \right) - \rho^{AB} \big\|_1 = 0 \implies \Pi_{x} \otimes \Lambda_{y} \left(\rho^{AB} \right) = \rho^{AB}.
\end{equation}
\end{proof}
%By similar reasoning, we see that $\mathbf{1} \otimes \Lambda_{(y,2)} \left(\rho^{AB} \right) = \rho^{AB}$. 

We now construct a classical strategy for the game $G$ in terms of the following joint probability distribution
\begin{eqnarray}
&&P^{JPD}\left(a_1,a_2,\ldots, a_{|\mathsf{X}|}, b_1,b_2,\ldots, b_{|\mathsf{Y}|} \big| x_1, x_2,\ldots,x_{|\mathsf{X}|}, y_1, y_2, \ldots, y_{|\mathsf{Y}|} \right) \nonumber \\
&&= \tr \left( \left(\Pi_{x_{|\mathsf{X}|}}^{a_{|\mathsf{X}|}} \cdot \ldots \cdot \Pi_{x_1}^{a_1} \otimes \Lambda_{y_{|\mathsf{Y}|}}^{b_{|\mathsf{Y}|}} \cdot \ldots \cdot \Lambda_{y_1}^{b_1} \otimes \mathbf{1} \right) \rho^{ABC} \right),
\end{eqnarray}
where we have chosen an ordering of the inputs such that $\mu_{X}(x_1) \geq \mu_{X}(x_2) \geq \ldots \geq \mu_{X}(x_{|\mathsf{X}|})$ and similarly $\mu_{Y}(y_1) \geq \mu_{Y}(y_2) \geq \ldots \geq \mu_{Y}(y_{|\mathsf{Y}|})$. The marginals of $P^{JPD}$ give the classical probabilities $P^{JPD}(a_k, b_l|x_k, y_l)$. We now complete the proof by showing that these marginals in fact reproduce the quantum probabilities.
% so that $\omega_{\mathcal{Q}}(G) = \omega_{\mathcal{C}}(G)$. 

Define the states $\rho^{AB}(k,l)$ as
\begin{eqnarray}
\rho^{AB}(k,l) := \left( \Pi_{x_{k-1}} \circ \ldots \circ \Pi_{x_{1}} \right) \otimes \left( \Lambda_{y_{l-1}} \circ \ldots \circ \Lambda_{y_{1}} \right) \rho^{AB},
\end{eqnarray}
with $\rho^{AB}(1,1) = \rho^{AB}$. 
Evidently, we have by Proposition \ref{prop:fixed-state} that the local measurements do not disturb the state $\rho^{AB}$ so that
\begin{equation}
\rho^{AB}(k,l) = \rho^{AB} \quad \forall k, l.
\end{equation}
Also, we have that the quantum probabilities
\begin{eqnarray}
P_{AB|XY}(a_k,b_l|x_k, y_l) &=& \tr\left( \Pi_{x_k} \otimes \Lambda_{y_l} \rho^{AB} \right) \nonumber \\
&=& \tr\left( \Pi_{x_k} \otimes \Lambda_{y_l} \rho^{AB}(k,l) \right) = P^{JPD}(a_k, b_l|x_k, y_l),
\end{eqnarray} 
which are exactly the classical probabilities arising from the joint probability distribution. Since the set of quantum probabilities in the marginal behavior $\{P_{AB|XY}(a,b|x,y)\}$ of the tripartite quantum behavior $\{P_{ABE|XYZ}(a,b,e|x,y,z)\}$ admit a joint probability distribution, the value that they achieve for the game $G$ cannot be super-classical, i.e., $\omega_{obs}(G) \leq \omega_{\mathcal{C}}(G)$ in \eqref{eq:guessing-game}. This shows that any behavior $\{P_{AB|XY}(a,b|x,y)\}$ that allows for perfect guessing of Alice-Bob's outcomes by Eve in this scenario must necessarily be local, i.e., no nonlocal quantum behavior exhibits bound randomness.

\end{proof} 

A consequence of Thm. \ref{thm:no-bound-rand} is in DI randomness amplification protocols, where the inputs are chosen using an $\epsilon$-SV source, so that 
\begin{equation}
  \left(\frac{1}{2} - \epsilon\right)^{\log_2 | X||Y|} \leq \nu_{XY}(x,y) \leq \left(\frac{1}{2} + \epsilon\right)^{\log_2 |X||Y|}.
\end{equation}
Therefore, in a DI amplification protocol such as in \cite{FWE+23, KAF20, CR12, BRG+16, GMT+13, RBH+16}, all measurement runs are considered as both Bell test runs and randomness generation runs. Since all measurement inputs are used for randomness generation, we see by Thm. \ref{thm:no-bound-rand} that any quantum nonlocal behavior is a sufficient resource for DI randomness amplification. This is in contrast to recent findings regarding spot-checking protocols for DI randomness expansion \cite{RYS24} and protocols for DI quantum key distribution \cite{FBL+21}.

%%%%%%%%%%%%%%%%%%%%%%%%%%%%%%%%%%%%%%%%%%%%%%%%%%%%%%%%%%%%%%%%%%%

\section{Detection Efficiency thresholds for randomness extraction vis-{\`a}-vis nonlocality detection}
\label{app:det-eff-I4422}

We have seen in the previous sections that for any Bell inequality (with any number of players, inputs and outputs), its violation is both a necessary and sufficient condition for device-independent randomness certification from the system, albeit in some cases one has to consider the average guessing probability and tailor the structure of the randomness certification protocol accordingly. 

\begin{figure}[htbp]
    \centering
    \includegraphics[width=0.7\linewidth]{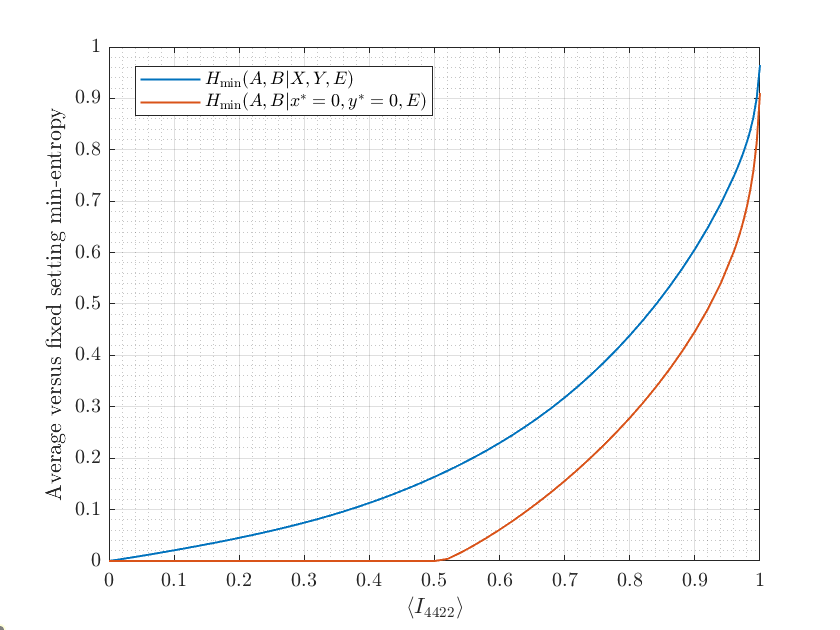}
    \caption{The horizontal axis represents the violation of the inequality $I_{4422}$, normalized such that its maximal quantum value equals $1$. The vertical axis shows the lower bound on certified randomness, quantified by the min-entropy - where we compare the average versus fixed-setting cases. Specifically, $H_{\min}(A,B|x^*=y^*=0,E)$ is computed for a fixed measurement setting $x^*=y^*=0$ (any fixed setting gives the same min-entropy for this inequality), while $H_{\min}(A,B|X,Y,E)$ is evaluated over four settings $x=y \in \{0,1,2,3\}$, each chosen with uniform probability $\mu_{XY}(x,y) = 1/4$. As can be seen, even though the $H_{\min}(A,B|x^*=y^*=0,E)$ drops to $0$ when $I_{4422} \in [0,0.499]$ as found in \cite{ZLH+25}, the average min-entropy for randomly chosen settings $x=y \in \{1,2,3,4\}$ is non-zero for any violation, i.e., $H_{\min}(A,B|X,Y,E) > 0$ for all $I_{4422} > 0$ for $\mu_{XY}(x,y) = 1/4$ for $x=y \in \{1,2,3,4\}$. 
    The numerical calculations are obtained using the NPA~\cite{NPA1,NPA2} hierarchy at level $2$ (all words of length two).}
    \label{fig:avg-Hmin-I4422}
\end{figure}

In this section, we extend these considerations to the detection efficiency thresholds necessary for the certification of randomness. In particular, in a recent Letter \cite{ZLH+25}, the specific $I_{4422}$ inequality (with $4$ inputs for each of Alice and Bob and $2$ outputs per input) was considered. This inequality has been highlighted for its ability to tolerate a detection efficiency as low as $61.8\%$~\cite{vertesi2010closing} in contrast to the Eberhard inequality which achieves the optimal detection efficiency threshold of $2/3$ in the $2$ input, $2$ output scenario. Given this low detection efficiency threshold for detecting nonlocality, this inequality might be of practical interest for experimental implementation in a DI randomness expansion or amplification protocol. However, in \cite{ZLH+25}, it was shown by means of a numerical calculation at the $1+AB+AAB$ level of the NPA hierarchy that the fixed setting min-entropy falls to zero for a detection efficiency below $90.6\%$ thereby rendering the $I_{4422}$ inequality much less desirable in practical randomness generation experiments. Explicitly, the $I_{4422}$ Bell inequality in the scenario where Alice and Bob measure binary observables $A_1, A_2, A_3, A_4$ and $B_1, B_2, B_3, B_4$ taking values in $\{0, 1\}$ is \cite{ZLH+25}
\begin{eqnarray}
    I_{4422} := I_{\text{CH}}^{(1,2;1,2)} + I_{\text{CH}}^{(3,4;3,4)} - I_{\text{CH}}^{(2,1;4,3)} - I_{\text{CH}}^{(4,3;2,1)} - p(A_2 = 1) - p(A_4 - 1) - p(B_2 = 1) - p(B_4 = 1) \leq 0,
\end{eqnarray}
where $I_{\text{CH}}^{(i,j;k,l)} = p(A_i = B_k = 1) + p(A_j = B_k = 1) + p(A_i = B_l = 1) - p(A_j = B_l = 1) - p(A_i = 1) - p(B_k = 1)$.

In \cite{ZLH+25}, it was found that when we consider the $I_{4422}$ inequality normalized such that its classical bound is $0$ and maximum quantum value is $1$, then the min-entropy $H_{\min}(A,B|x^*,y^*,E) = 0$ for $I_{4422} < 0.5$ for any fixed $(x^*,y^*) \in \{1,2,3,4\}^2$. That is, we have that the fixed-setting guessing probability $P^{(q)}_g(\textbf{p}, (x^*, y^*)) = 1$ when $I_{4422} < 0.5$ for any fixed $(x^*, y^*)$. The detection efficiency threshold to achieve $I_{4422} > 0$ in quantum theory is famously equal to $61.8\%$ however the detection efficiency threshold required to achieve $I_{4422} > 0.5$ was found to be $90.6\%$ thus calling into question the effectiveness of this inequality in DI randomness certification. 

In Fig. \ref{fig:avg-Hmin-I4422}, we compute the average min-entropy 
$H_{\min}(A,B|X,Y,E)$ averaged over uniformly random $x=y \in \{1,2,3,4\}$  (obtained by computing  $-\log_2 P^{(q)}_{g}(\textbf{p}, \mu_{XY})$ for $\mu_{XY}(x,y) = 1/4$ for $x = y \in \{1,2,3,4\}$) and show explicitly that the inequality certifies randomness on average for any violation, thus restoring the attractiveness of this inequality for use in DI randomness certification protocols.

In Fig. \ref{fig:avg-Hmin-detectionefficiency}, we reinforce the point by computing the $H_{\min}(A,B|X,Y,E)$ for $\mu_{XY}(x,y) = 1/4$ for any $x=y \in \{1,2,3,4\}$ as a function of the violation of $I_{4422}(\eta)$~\cite{vertesi2010closing} where the detection efficiency $\eta$ is explicitly incorporated into the inequality. Specifically, we consider the inequality~\cite{vertesi2010closing}
\begin{equation}
    I_{4422}(\eta) := I_{\text{CH}}^{(1,2;1,2)}(\eta) + I_{\text{CH}}^{(3,4;3,4)}(\eta) - I_{\text{CH}}^{(2,1;4,3)}(\eta) - I_{\text{CH}}^{(4,3;2,1)}(\eta) - \left[p(A_2 = 1) - p(A_4 - 1) - p(B_2 = 1) - p(B_4 = 1)\right]/\eta \leq 0,
\end{equation}
where
\begin{equation}
    I_{\text{CH}}^{(i,j;k,l)} (\eta) = p(A_i = B_k = 1) + p(A_j = B_k = 1) + p(A_i = B_l = 1) - p(A_j = B_l = 1) - \left[p(A_i = 1) + p(B_k = 1) \right]/\eta.
\end{equation}
As can be seen from the plot, the optimal strategies depend explicitly on the detection efficiency $\eta$ in question, but it is always the case that $H_{\min}(A,B|X,Y,E) > 0$ for any $\eta > 61.8\%$. 

In conclusion, we see that a consideration of the average guessing probability restores the desirability of a nonlocal test via the inequality $I_{4422}$ for DI randomness certification in practical experiments such as those in \cite{ZLH+25} where the detection efficiency thresholds play a crucial role. We also remark that the inequality $I_{4422}$ should also be considered a good candidate for DI randomness expansion protocols, provided we tailor the structure of the protocol to be one without spot-checking. For instance, one may modify the protocol to utilize a small amount of randomness to vary the input settings used in the generation rounds. One may also consider protocols when the input randomness is recycled or a protocol using a source of public randomness (such as NIST's randomness beacon) in which case the focus is on turning public randomness into private quantum-certified randomness \cite{BRC23}. 

\begin{figure}[htbp]
    \centering
    \includegraphics[width=0.7\linewidth]{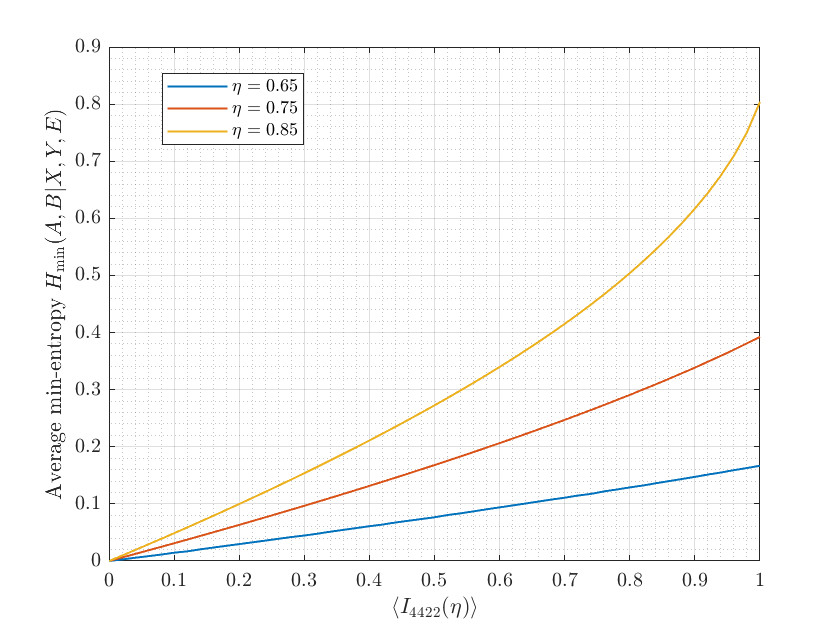}
    \caption{Plot of the average min-entropy $H_{\min}(A,B|X,Y,E)$ for $\mu_{X,Y}(x,y) = 1/4$ for any $x=y \in \{1,2,3,4\}$ as a function of the violation of $I_{4422}(\eta)$ where the detection efficiency $\eta$ is explicitly incorporated into the inequality. 
    The horizontal axis corresponds to the violation of the inequality $I_{4422}(\eta)$, again normalized to have a maximum quantum value of $1$. The vertical axis represents the min-entropy $H_{\min}(A,B|X,Y,E)$, where the inputs $x=y \in \{1,2,3,4\}$ are uniformly distributed. The optimal strategies depend specifically on the detection efficiency $\eta$ in question, but it is always the case that $H_{\min}(A,B|X,Y,E) > 0$ for any $\eta > 61.8\%$. }
    \label{fig:avg-Hmin-detectionefficiency}
\end{figure}

%%%%%%%%%%%%%%%%%%%%%%%%%%%%%%%%%%%%%%%%%%%%%%%%%%%%%%%%%%%%%%%%%%%

\section{Average Guessing Probability as a measure of Quantum Nonlocality}
\label{app:pgavg-nl-meas}

In this Appendix, we show that the average device–independent guessing probability against a quantum adversary constitutes a valid nonlocality measure. Specifically, the corresponding min–entropy is faithful (it vanishes exactly on local behaviors) and monotonic under single-copy WCCPI operations (wirings and classical communication performed prior to the inputs). Recall that WCCPI operations \cite{REAM12} consist of wirings and classical communication outside the measurement phase, and they serve to distinguish local correlations (which can be generated by WCCPI operations) from nonlocal correlations (which cannot).

A bipartite behavior is denoted by $\{p_{AB|XY}(a,b|x,y)\}$.  
Let $\mathcal{L}$ denote the local set, and let $\mathcal{Q}^{trip}$ ($\mathcal{NS}^{trip}$) denote the set of tripartite quantum (no-signalling) extensions compatible with $\{p_{AB|XY}\}$, as used to define device-independent (DI) guessing probabilities in App.~\ref{sec:DIGP}. For an input distribution $\mu_{XY}(x,y)$ on $X \times Y$, the average DI guessing probability against a quantum adversary is given by Eqs.~\eqref{eq:gpii}:
\begin{equation}
P^{(q)}_{\text{g}}(\mathbf p,\mu_{XY}) 
= \max_{\,\{p_{ABE|XYZ}\}\in\mathcal Q^{trip}}
\sum_{x,y}\mu_{XY}(x,y)\sum_{a,b} 
p_{ABE|XYZ}\bigl(a,b,e{=}(a,b)\,\big|\,x,y,z{=}(x,y)\bigr),
\label{eq:prev obj}
\end{equation}
subject to the constraint $\sum_e p_{ABE|XYZ}(a,b,e|x,y,z) = p_{AB|XY}(a,b|x,y)$. We quantify nonlocality via the min-entropy:
\begin{align}
\mathcal H_{\min}(\mathbf p, \mu_{XY}) 
= -\log_2 P^{(q)}_{\text{g}}(\mathbf p,\mu_{XY}).
\end{align}

We adopt the single-copy WCCPI operation framework $\mathcal W$ from \cite{DVJI14}. In a fixed measurement setting $(|X|,|Y|,|A|,|B|)$, Alice’s (Bob’s) input alphabet is 
$\{x_1, x_2, \ldots, x_{|X|}\}$ ($\{y_1, y_2, \ldots, y_{|Y|}\}$), and her (his) output alphabet is 
$\{a_1, a_2, \ldots, a_{|A|}\}$ ($\{b_1, b_2, \ldots, b_{|B|}\}$). Throughout, we restrict attention to single-copy manipulations of the resource, meaning that the parties cannot access multiple instances of the behavior jointly. The single-copy WCCPI operations that are given for free to both parties in the DIQIP scenario (device-independent quantum information processing) are mixtures of:
(i) relabelings $\mathcal R$ of inputs and outputs,
(ii) output coarse-grainings $\mathcal G$,
(iii) input substitutions $\mathcal K$, and
(iv) mixing with local behaviors.  

A nonlocality measure $\mathcal N$ must satisfy the following:  
\begin{itemize}
\item[(N1)] \emph{Faithfulness:} $\mathcal N(\mathbf p) = 0$ for all $\mathbf p \in \mathcal L$, and $\mathcal N(\mathbf p) > 0$ for any nonlocal $\mathbf p \notin \mathcal L$ (equivalently, $\mathcal N(\mathbf p) = 0$ iff $\mathbf p \in \mathcal L$).  
\item[(N2)] \emph{Monotonicity:} $\mathcal N(\mathcal W(\mathbf p)) \leq \mathcal N(\mathbf p)$ for every free single-copy WCCPI operation $\mathcal W$ (and for convex mixtures thereof).  
\end{itemize}

\begin{lem}
The average device-independent guessing probability against a quantum adversary $H_{\min}(\textbf{p}, \mu_{XY})= -\log_2 P^{(q)}_{\text{g}}(\mathbf p,\mu_{XY})$ for complete-support $\mu_{XY}$  is monotonic under single-copy WCCPI operations $\mathcal W$.
\end{lem}
    
\begin{proof}

Equivalently to (N2), since the min–entropy is defined as $\mathcal H_{\min}=-\log_2 P_{\text{g}}^{(q)}$, it suffices to prove that the device–independent guessing probability $P_{\text{g}}^{(q)}$ is non-decreasing under any single-copy WCCPI operation $\mathcal W$. In the DIQIP scenario, all free WCCPI operations are implemented outside the measurement phase (before inputs are chosen and after outcomes are produced), allowing unrestricted local processing. However, Eve can still learn the inputs to the device and perform her measurements accordingly. According to the definitions in App.\ref{sec:DIGP}, Eve learns the actual inputs $(x,y)$ and optimizes her strategy after the operations have been applied. Consequently, verifying (N2) amounts to showing that, for each class of free operations $\mathcal W$ in the WCCPI framework, Eve’s optimal guessing probability $P^{(q)}_{\text{g}}$ on $\mathcal W(\mathbf p)$ is at least as large as that on $\mathbf p$.

\noindent\emph{(i) Relabeling $\mathcal R$.}
Let $\mathcal R$ act as a bijection (permutation) on the input and output alphabets: $(x,y) \mapsto (\tilde x,\tilde y)$ and $(a,b) \mapsto (\tilde a,\tilde b)$. Since the optimization is taken over all $p_{ABE|XYZ} \in \mathcal Q^{trip}$, and both the constraints and the objective are invariant under relabeling of the measurement settings $(x,y)$ and the recorded outputs $(a,b)$, a change of variables yields the following equivalence:
\begin{align}
P^{(q)}_{\text{g}}(\mathcal R(\mathbf p), \mu_{XY}(\tilde{x},\tilde{y})) 
= P^{(q)}_{\text{g}}(\mathbf p, \mu_{XY}(x,y)).
\end{align}

\noindent\emph{(ii) Output coarse-graining $\mathcal G$.}
Let $A'$ ($B'$) be a subset of the output set $A$ ($B$) on Alice’s (Bob’s) side, and define the coarse-grained behavior $\mathbf p' = \mathcal G(\mathbf p)$ by merging all outputs in $A'$ ($B'$) into a single new output $a'_{\text{new}} \in A'$ ($b'_{\text{new}} \in B'$). The new probability distribution is given by $p'_{AB|XY}(a'_{\text{new}}, b'_{\text{new}}|x,y) = \sum_{a \in A',\, b \in B'} p_{AB|XY}(a,b|x,y)$. The fixed measurement setting condition and the input distribution remain unchanged, since we can regard the new behavior as having the same number of outputs as before, but with some of them occurring with zero probability.

Now consider the same optimization problem with the new tripartite behavior $p'_{ABE|XYZ}(a',b',e'|x,y,z) \in \mathcal Q^{trip}$ in Eq.\eqref{eq:gpii}. The new objective function becomes
\begin{align}
\max_{p'_{ABE|XYZ}\in \mathcal Q^{trip}} \sum_{x,y} \mu_{XY}(x,y) \sum_{a' \in A_{\text{cg}},\, b' \in B_{\text{cg}}} 
p'_{ABE|XYZ}(a',b',e'=(a',b')|x,y,z),
\label{eq:new obj ocg}
\end{align}
where the coarse-grained output sets are 
$A_{\text{cg}} = \overline{A'} \cup \{a'_{\text{new}}\}$ and 
$B_{\text{cg}} = \overline{B'} \cup \{b'_{\text{new}}\}$, with $\overline{A'}$ and $\overline{B'}$ denoting the complements of $A'$ and $B'$, respectively.  
The optimal guessing strategy satisfies:
\begin{align}
\sum_{e'} p'_{ABE|XYZ}(a'_{\text{new}},b'_{\text{new}},e'|x,y,z) 
&= \sum_{e} \sum_{a \in A',\, b \in B'} p_{ABE|XYZ}(a,b,e|x,y,z), \\
p'_{ABE|XYZ}(a'_{\text{new}},b'_{\text{new}},e' =(a'_{\text{new}},b'_{\text{new}})|x,y,z) 
&= \sum_{a \in A',\, b \in B'} p_{ABE|XYZ}(a,b,e=(a,b)|x,y,z).
\end{align}

Therefore, the original optimal solution for the average guessing probability is also a feasible solution to the new optimization problem in Eq.\eqref{eq:new obj ocg}, with the $\overline{A'}$ part unchanged and the additional output $a'_{\text{new}}$ behaving as shown above under the fixed input distribution $\mu_{XY}$. Since this is a maximization problem, the new guessing probability cannot be smaller than the original one before applying the output coarse-graining operation $\mathcal G$. Hence,
\begin{align}
P^{(q)}_{\text{g}}(\mathcal G(\mathbf p), \mu_{XY}) 
\;\;\geq\;\; P^{(q)}_{\text{g}}(\mathbf p, \mu_{XY}).
\end{align}

\noindent\emph{(iii) Input substitution $\mathcal K$.}  
In the fixed measurement setting, an input substitution operation replaces one measurement setting $x_i \in X$ with another pre-existing setting $x_j \in X$ for $i \neq j$. Define the input substitution map $\mathcal K_{i\leftarrow j}$ by
\begin{equation}
p'_{AB|XY}(a,b|x,y)=
\begin{cases}
p_{AB|XY}(a,b|x_j,y),& x=x_i,\\
p_{AB|XY}(a,b|x,y),& x\neq x_i,
\end{cases}
\end{equation}
Operationally, this can be viewed as the composition of an input shortening operation $\mathcal S$ (removing $x_i$) followed by a correlated input largening operation $\mathcal E_c$ (duplicating $x_j$). After the substitution, the input distribution $\mu_{XY}(x,y)$ is unchanged. Relative to the objective function in the original optimization problem, the terms affected by the substitution are
\begin{equation}
\begin{split}
    \mu_{XY}(x_i,y)&\sum_{a,b}p'_{ABE|XYZ}(a,b,e|x_i,y,z) 
    + \mu_{XY}(x_j,y)\sum_{a,b}p'_{ABE|XYZ}(a,b,e|x_j,y,z) \\
    &= \;  \mu_{XY}(x_i,y)\sum_{a,b}p_{ABE|XYZ}(a,b,e|x_j,y,z) 
    + \mu_{XY}(x_j,y)\sum_{a,b}p'_{ABE|XYZ}(a,b,e|x_j,y,z) \\
    &= \; \sum_{x \in \{x_i,x_j\}, y} \mu_{XY}(x,y) \sum_{a,b}p'_{ABE|XYZ}(a,b,e|x_j,y,z)
\end{split}
\label{eq:prev terms}
\end{equation}
The new optimization objective function is therefore
\begin{equation}
\begin{split}
& P^{(q)}_{\text{g}}(\mathcal K(\mathbf p),\mu_{XY}) 
= \max_{p'} \sum_{x,y} \mu_{XY}(x,y) \sum_{a,b} p'_{ABE|XYZ}(a,b,e|x,y,z) \\
& = \max_{p'} \Bigl( \sum_{x \notin \{x_i,x_j\},\,y} \mu_{XY}(x,y) \sum_{a,b} p'_{ABE|XYZ}(a,b,e|x,y,z) + \sum_{x \in \{x_i,x_j\}, y} \mu_{XY}(x,y) \sum_{a,b} p'_{ABE|XYZ}(a,b,e|x_j,y,z) \Bigr) \\
& \overset{\mathcal R}{=} \max_{p'} \Bigl( \sum_{x \notin \{x_i,x_j\},\,y} \mu_{XY}(x,y) \sum_{a,b} p'_{ABE|XYZ}(a,b,e|x,y,z) + \sum_{x \in \{x_i,x_j\}, y} \mu_{XY}(x,y) \sum_{a,b} p'_{ABE|XYZ}(a,b,e|x_i,y,z) \Bigr)\\
& = P^{(q)}_{\text{g}}(\mathcal R(\mathcal K(\mathbf p)), \mu_{XY}(\tilde{x},\tilde{y})),
\end{split}
\end{equation}
where the last step follows from applying the relabeling operation $\mathcal R_{(x_i \leftrightarrow x_j)}$ to the behavior, which preserves the equivalence relation. From the intermediate steps above, the input weights are unchanged under relabeling of $x_i$ and $x_j$ equaling to the combined weight $\mu_{XY}(x_i,y) + \mu_{XY}(x_j,y)$. Hence, we conclude that in the maximization problem
\begin{align}
\sum_{a,b} p'_{ABE|XYZ}(a,b,e|x_j,y,z) 
= \sum_{a,b} p'_{ABE|XYZ}(a,b,e|x_i,y,z),
\end{align}
which means that the original optimal solution of the optimization remains a solution for the new optimization:
\begin{align}
\sum_{a,b} p'_{ABE|XYZ}(a,b,e|x_j,y,z) 
= \max \, \left\{ \sum_{a,b} p_{ABE|XYZ}(a,b,e|x_i,y,z), \;\; \sum_{a,b} p_{ABE|XYZ}(a,b,e|x_j,y,z) \right\}.
\end{align}
By the same argument as in case (ii), since this is a maximization problem, the guessing probability cannot decrease under input substitution:
\begin{align}
P^{(q)}_{\text{g}}(\mathcal K(\mathbf p),\mu_{XY}) 
\;\;\geq\;\;
P^{(q)}_{\text{g}}(\mathbf p,\mu_{XY}).
\end{align}

\noindent\emph{(iv) Mixing with local behaviors.}  
Consider $\mathbf p' = \lambda \cdot \mathbf p + (1-\lambda)\cdot \ell$ with $\ell \in \mathcal L$ and $\lambda \in [0,1]$.  
The objective function in the new optimization problem can be written as
\begin{align}
    & P^{(q)}_{\text{g}}(\mathbf p', \mu_{XY}) = \max_{p'_{ABE|XYZ}} \sum_{x,y} \mu_{XY}(x,y) \sum_{a,b} p'_{ABE|XYZ}(a,b,e=(a,b)|x,y,z) \\
    & = \max_{p'_{ABE|XYZ}} \sum_{x,y} \mu_{XY}(x,y) \sum_{a,b} \left(\lambda \cdot p_{ABE|XYZ}(a,b,e=(a,b)|x,y,z) + (1-\lambda) \cdot p_{\ell}(a,b,e=(a,b)|x,y,z)\right) \\
    & = \max_{p'_{ABE|XYZ}} \Bigl(\lambda \cdot \sum_{x,y} \mu_{XY}(x,y) \sum_{a,b} (p_{ABE|XYZ}(a,b,e=(a,b)|x,y,z) +  (1-\lambda) \cdot \sum_{x,y} \mu_{XY}(x,y) \sum_{a,b} p_{\ell}(a,b,e=(a,b)|x,y,z))\Bigr)
\end{align}
The constraint on marginal consistency is analogous to Eq.\eqref{eq:prev obj},
\begin{align}
    \mathbf p'_{AB|XY}(a,b|x,y) & = \sum_e p'_{ABE|XYZ}(a,b,e|x,y,z)  \\
    & = \sum_e (\lambda \cdot p_{ABE|XYZ}(a,b,e|x,y,z) + (1-\lambda) \cdot p_{\ell}(a,b,e|x,y,z)) = \lambda \cdot \mathbf p _{AB|XY} + (1-\lambda) \cdot \ell
\end{align}
The new optimization problem is also subject to the constraint 
$\{p'(a,b,e|x,y,z)\} \in \mathcal Q^{\text{trip}}$.  
For any local behavior $\ell$, the adversary can always adopt a strategy that perfectly predicts Alice’s and Bob’s outcomes once the inputs $(x,y)$ are revealed.  
The difference compared to the original objective function in Eq.~\eqref{eq:prev obj} is that certain behavior $p_{ABE|XYZ}$ is replaced by the local behavior $p_{\ell}$ with weight $(1-\lambda)$, while the $\lambda$-weighted part of the optimization remains unchanged.  
This directly yields the relation
\begin{align}
P^{(q)}_{\text{g}}(\mathbf p', \mu_{XY}) 
\;\;\geq\;\;
P^{(q)}_{\text{g}}(\mathbf p, \mu_{XY}).
\end{align}
Hence, mixing with local behaviors cannot reduce the device–independent guessing probability, and therefore cannot increase the associated nonlocality measure.

We remark that two additional operations are sometimes listed: output unfolding $\mathcal U$ (splitting one outcome into different ones) and uncorrelated input enlargement $\mathcal E_u$ (local measurement inputs). As discussed in \cite{DVJI14}, in the fixed measurement setting scenario both operations can be expressed as combinations of the identity operation and output coarse-graining (i.e. $\mathbf p' = (1-\lambda)\cdot \mathbf p + \lambda \cdot \mathcal G(\mathbf p)$). Their effect on the device–independent guessing probability is therefore already subsumed by the arguments given above in~(ii).

Combining the four properties above, we obtain that 
\begin{align}
\mathcal H_{\text{min}}(\mathcal W(\mathbf p),\mu_{XY}) \;\; \leq \;\; \mathcal H_{\text{min}}(\mathbf p,\mu_{XY}),
\end{align}
and hence $\mathcal H_{\text{min}}(\textbf{p}, \mu_{XY})$ obeys (N2).
\end{proof}

\begin{lem}
    Among the guessing probability models defined in App.\ref{sec:DIGP}, the average guessing probability against a quantum adversary defines a proper nonlocality measure: it is faithful and monotonic under all single copy WCCPI operations.
\end{lem}

\emph{Fixed-setting (quantum) guessing probability is not faithful.}  
For the fixed pair $(x^*,y^*)$ defined in Eq.\ref{eq:gpi}, there exist nonlocal quantum behaviors $\mathbf p \in \mathcal Q \setminus \mathcal L$ such that $P^{(q)}_{\text{g}}(\mathbf p, (x^*, y^*)) = 1$ (Lemma~\ref{lem:1}). Hence, certain candidate measures $\mathcal H_{\min}^{(q)}(\mathbf p) = -\log_2 P^{(q)}_{\text{g}}(\mathbf p, (x^*, y^*))$ fail to satisfy (N1).

\emph{Average (quantum) guessing probability is faithful.}  
If $\mathbf p \in \mathcal L$, then Eve can share the local variable and deterministically output $(a,b)$ given $(x,y)$. Hence $P^{(q)}_{\text{g}}(\mathbf p,\mu_{XY}) = 1$ and $\mathcal H_{\min}(\mathbf p,\mu_{XY}) = 0$. Conversely, for nonlocal quantum behaviors $\mathbf p \in \mathcal Q \setminus \mathcal L$, Theorem.\ref{thm:no-bound-rand} rules out perfect prediction once $\mu_{XY}$ is known to Eve; in particular, $P^{(q)}_{\text{g}}(\mathbf p,\mu_{XY}) < 1 \text{ and }\mathcal H_{\min}(\mathbf p,\mu_{XY}) > 0$. Thus, $\mathcal H_{\min}(\mathbf p,\mu_{XY})$ satisfies (N1).

\emph{No-signalling adversary: both fixed and average NS guessing probability are not faithful.}  
Let $P^{(ns)}_{\text{g}}$ denote Eve’s guessing probability when the adversary has access to general no-signalling side information. There exist nonlocal behaviors $\mathbf p \in \mathcal{NS} \setminus \mathcal L$ for which a no-signalling Eve can guess perfectly on any pair of input settings (see the remark in App.\ref{app:no-bound-rand}), i.e., $P^{(ns)}_{\text{g}}(\mathbf p, (x^*,y^*)) = 1 \quad \text{for all } (x^*,y^*)$.
Furthermore, for any input distribution $\mu_{XY}$, one has $P^{(ns)}_{\text{g}}(\mathbf p, \mu_{XY}) 
= \sum_{x^*,y^*} \mu_{XY}(x^*,y^*)\, P^{(ns)}_{\text{g}}(\mathbf p, (x^*,y^*)) 
\; \text{(Lemma \ref{lem:2})}$,
and hence $-\log_2 P^{(ns)}_{\text{g}}(\mathbf p,\mu_{XY}) = 0
\; \text{although } \mathbf p \notin \mathcal L$.
Thus, the guessing probability against a no-signalling adversary is not faithful and cannot serve as a proper nonlocality measure.

We therefore conclude that, for quantum adversaries and for complete-support input distributions $\mu_{XY}$, the average guessing probability is a valid nonlocality measure: it is faithful, vanishing if and only if the behavior is local, and monotonic under single-copy WCCPI operations. In contrast, the fixed-setting quantum guessing probability and the guessing probability against a no-signalling adversary are not faithful, as they assign the value $0$ to certain nonlocal behaviors, and thus cannot be regarded as proper nonlocality measures.

%%%%%%%%%%%%%%%%%%%%%%%%%%%%%%%%%%%%%%%%%%%%%%%%%%%%%%%%%%%%%%%%%%%

\section{Average Guessing Probability in the CHSH Bell test}
\label{app:Pgavg-CHSH}

In this Appendix, we compute an analytical expression for the average guessing probability by a quantum Eve of the measurement outputs of one player Alice in a CHSH Bell test. Here, the averaging is done over the input distribution $\mu_X$ which represents Eve's classical side information about Alice's inputs. 

Consider a CHSH Bell test in which Alice measures binary observables $A_0, A_1 \in \{\pm 1\}$ and similarly Bob measures binary observables $B_0, B_1$ chosen with uniform probabilities in the test rounds. Suppose that Alice and Bob observe a value $I_{obs} \in [2, 2\sqrt{2}]$ for the CHSH Bell expression given as
$I_{CHSH} = \langle A_0 B_0 \rangle + \langle A_0 B_1 \rangle + \langle A_1 B_0 \rangle - \langle A_1 B_1 \rangle$, with the correlators defined as $\langle A_i B_j \rangle:= \sum _{a,b\in\{\pm 1\}} ab\cdot p_{AB|XY}(a,b|A_i,B_j).$
% \begin{equation}\label{eq_corr}
%     \langle A_i B_j \rangle = p_{AB|XY}(0,0|A_i,B_j) + p_{AB|XY}(1,1|A_i,B_j) - p_{AB|XY}(0,1|A_i,B_j) - p_{AB|XY}(1,0|A_i,B_j).
% \end{equation}

Suppose now that during the randomness generation rounds Alice chooses her inputs $A_0$ and $A_1$ with probabilities $p$ and $1-p$ respectively for some $p \in \left[ \frac{1}{2}, 1 \right]$. We are then interested in the analytical form of the average guessing probability $P_g^{(q, A)}(I_{obs}, \{p,1-p\})$ of Alice's measurement outputs by a quantum adversary Eve. This local average guessing probability is defined following Appendix \ref{sec:DIGP} as
\begin{eqnarray}
\label{eq:gp-al-Bell}
 P^{(q, A)}_{g}(I_{obs}, \{p,1-p\}) &=& \max_{\{p_{ABE|XYZ}\}} \sum_{a= \pm 1} p \cdot p_{AE|XZ}(a, e=a|A_0, z=0) + (1-p) \cdot p_{AE|XZ}(a, e=a|A_1, z= 1) \nonumber \\ 
   && \qquad \qquad \;\; I_{CHSH}\left(p_{AB|XY} \right) = I_{obs} \; \; \; \; \nonumber \\
   && \; \qquad \qquad \; \{ p_{ABE|XYZ} \} \in \mathcal{Q}^{trip}.
\end{eqnarray}
Here, $p_{AE|XZ}$ and $p_{AB|XY}$ are marginals of the tripartite quantum behavior $p_{ABE|XYZ}$ and $I_{obs} \in [2, 2 \sqrt{2}]$ is some observed value by Alice-Bob of the CHSH parameter $I_{CHSH}$. 

\begin{thm}
Consider a CHSH Bell test in which Alice and Bob choose their inputs $A_0, A_1$ and $B_0, B_1$ with uniform probabilities respectively in the test rounds and observe a CHSH value $I_{obs}\in[2,2\sqrt{2}]$. Suppose that during the randomness generation rounds Alice chooses her inputs $A_0, A_1$ with respective probabilities $p, 1-p$ for some $p \in [\frac{1}{2}, 1]$. Then the corresponding average guessing probability of Alice's outputs by a quantum Eve satisfies
\begin{equation}
    P^{(q,A)}_{g}(I_{obs},\{p,1-p\})
    \geq \frac12 + \frac{\sqrt{p^2 + (1-p)^2}}{2} \, f_{\phi_p}(I_{obs}),
\end{equation}
where $\phi_p := 2\arctan\!\left(\frac{1-p}{p}\right)$ and
$f_{\phi}(I_{obs}):=\frac{\sqrt{\frac{\left(4-\alpha_{\min,\phi}^2\right)\left(32-2\alpha_{\min,\phi}^4\cos^2 \phi\right)}{16-8 \alpha_{\min,\phi}^2+\alpha_{\min,\phi}^4\cos^2 \phi}}-I_{obs}}{\alpha_{\min,\phi}}$. Here $\alpha_{\min,\phi}$ is the real root in $[0,\alpha_\phi]$ of
$h_{\phi}(t) = I_{obs}$, with $h_{\phi}(t) :=
\frac{4\sqrt{2}\,[32(8-8t^2+t^4)+2t^4\cos^2\phi(32-8t^2+t^4)-t^8\cos^4\phi]}
{\sqrt{(t^2-4)(t^4\cos^2\phi-16)}\,(16-8t^2+t^4\cos^2\phi)^{3/2}}$, and $\alpha_{\phi}$ the real solution in $[0,2]$ of $\frac{\alpha^2 \sin \phi (96-16\alpha^2-2\alpha^4\cos^2\phi)}
{(4-\alpha^2)(32-16\alpha^2+2\alpha^4\cos^2\phi)} = 1$.
\end{thm}

\begin{proof}
%We consider device-independent randomness certification protocols where the CHSH inequality $I_{CHSH}$ is used for the game rounds. In addition, for the randomness generation round, the measurements $A_0$ and $A_1$ are chosen with probabilities $p$ and $1-p$, respectively. Since applying the local transformation $A_0 \leftrightarrow A_1$ to the CHSH Bell expression yields an equivalent Bell expression, it suffices to analyze the region $p \in [\frac{1}{2},1]$. 
%The tight bound on the guessing probability as a function of the Bell violation can be derived from the tilted version of the Bell expression~\cite{}. In this section, we adopt this approach and 
In this proof, we use the double-tilted CHSH expression studied in~\cite{mikos2023extremal} to derive the average guessing probability as a function of the observed CHSH value $I_{obs}$.

For parameters $\alpha \in [0,2)$ and $\phi \in [0,\frac{\pi}{2}]$, define the double-tilted CHSH expression~\cite{mikos2023extremal} as  
\begin{equation}
	I_{CHSH}^{(\alpha,\phi)}:=\alpha \cos\left(\frac{\phi}{2} \right) A_0 +\alpha \sin \left(\frac{\phi}{2} \right) A_1 + I_{CHSH}.
\end{equation}
The maximum classical value of this Bell expression is seen to be $I_{L}^{(\alpha, \phi)} = \sqrt{2} \alpha \sin\left(\frac{\pi + 2 \phi}{4} \right)+ 2$. In \cite{mikos2023extremal}, it was observed that the maximum quantum value of $I_{CHSH}^{(\alpha,\phi)}$ is greater than the classical value only for $\alpha$ in the region $[0, \alpha_{\phi})$. Here, for any $\phi \in (0,\frac{\pi}{2}]$, $\alpha_{\phi}$ is defined as the real solution in $[0,2]$ of the following equation
\begin{equation}
\label{def:alpha-phi}
	\frac{\alpha^2 \sin \phi\left(96-16 \alpha^2-2\alpha^4\cos^2 \phi\right)}{\left(4-\alpha^2\right)\left(32-16 \alpha^2+2\alpha^4\cos^2 \phi\right)} =1.
\end{equation}
Furthermore, when $\phi = 0$, we set $\alpha_{\phi} = 2$. 
%we have that the quantum value can exceed the local value for $\alpha$ in the region $[0,\alpha_\phi)$, where $\alpha_\phi \in [0,2]$. Specifically, when $\phi=0$, we have $\alpha_0=2$, and for $\phi \in (0,\frac{\pi}{2}]$, $\alpha_\phi$ is the solution of the following equation:
%\begin{equation}
%	\frac{\alpha^2 \sin \phi\left(96-16 \alpha^2-2\alpha^4\cos^2 \phi\right)}{\left(4-\alpha^2\right)\left(32-16 \alpha^2+2\alpha^4\cos^2 \phi\right)} =1.
%\end{equation}
Specifically, for $\phi \in [0,\frac{\pi}{2}]$ and $\alpha \in [0,\alpha_\phi]$, the maximum quantum value of the double-tilted CHSH expression $I_{CHSH}^{(\alpha,\phi)}$ was found to be 
\begin{equation}\label{eq_Iq_dtchsh}
	I_Q^{(\alpha, \phi)} = \sqrt{\frac{\left(4-\alpha^2\right)\left(32-2\alpha^4\cos^2 \phi\right)}{16-8 \alpha^2+\alpha^4\cos^2 \phi}} .
\end{equation}
When the observed value of the CHSH expression $I_{CHSH}$ is $I_{obs} \in [2,2\sqrt{2}]$, we obtain the upper bound for the quantity $\cos\frac{\phi}{2}\<A_0\> + \sin\frac{\phi}{2}\<A_1\>$ for any $\phi \in [0, \frac{\pi}{2}]$ as:
\begin{equation}\label{db_tilted_bound}
	\cos\frac{\phi}{2} \<A_0\> + \sin\frac{\phi}{2} \<A_1\> \leq \min_{\alpha \in [0,\alpha_\phi)} \frac{I_Q^{(\alpha, \phi)} - I_{obs}}{\alpha}.
\end{equation}
By setting for any fixed $\phi \in [0, \frac{\pi}{2}]$,
\begin{equation}
\frac{\partial}{\partial \alpha} \frac{I_Q^{(\alpha, \phi)} - I_{obs}}{\alpha} = 0,
\end{equation}
and solving the resulting equation, we obtain the $\alpha$ that achieves the minimum in the RHS of \eqref{db_tilted_bound} as $\alpha_{\min, \phi} = h_{\phi}^{-1}(I_{obs})$, where
\begin{equation}\label{eq_h_phi}
	h_{\phi}(t) :=\frac{4\sqrt{2}\left[ 32\left(8-8 t^2+ t^4\right)+2 t^4\cos^2\phi\left(32-8 t^2+t^4\right)-t^8\cos^4\phi \right]}{\sqrt{(t^2-4)(t^4\cos^2\phi-16)}\,(16-8 t^2+ t^4\cos^2\phi)^{3/2}}.
\end{equation}
Substituting $\alpha_{\min, \phi}$ back into the RHS of Eq.~\eqref{db_tilted_bound} and denoting the resulting RHS as $f_{\phi}(I_{obs})$, we obtain the upper bound for the quantity $\cos\frac{\phi}{2}\<A_0\> + \sin\frac{\phi}{2}\<A_1\>$ for any $\phi \in [0, \frac{\pi}{2}]$ as:
\begin{equation}
	\cos\frac{\phi}{2} \<A_0\> + \sin\frac{\phi}{2} \<A_1\> \leq f_{\phi}(I_{obs}).
\end{equation}
%This bound is tight, as the equality is achieved by the optimal quantum strategy of the double-tilted CHSH expression for the corresponding parameters. 
To study the average guessing probability of Alice's outputs with input distribution $\{p, 1-p\}$ for any $p \in \left[\frac{1}{2}, 1 \right]$, we set
\begin{equation}
	\phi_p :=2\arctan\left(\frac{1-p}{p}\right).
\end{equation}
Substituting $\phi_p$ into $f_{\phi}(I_{obs})$ results in the function $f_{\phi_p}(I_{obs})$, and then we have the bound
\begin{equation}\label{db_tilted_bound_2}
	p \<A_0\> + (1-p) \<A_1\> \leq f_{\phi_p}(I_{obs})\sqrt{p^2 + (1-p)^2}.
\end{equation}
Furthermore, we check that this bound is tight by minimizing the difference of the RHS and LHS of the above expression under the constraint that the observed value of $I_{CHSH}$ is $I_{obs}$, relaxing the constraint set using the NPA semidefinite programming hierarchy and using SDP duality. 

Moreover, applying the local transformation $A_i \leftrightarrow -A_i,\,B_j \leftrightarrow -B_j$ for all $i,j\in\{0,1\}$ also yields a CHSH Bell expression, giving a similar tight bound
\begin{equation}\label{db_tilted_bound_3}
	-p \<A_0\> - (1-p) \<A_1\> \leq f_{\phi_p}(I_{obs})\sqrt{p^2 + (1-p)^2}.
\end{equation}
%We denote the average guessing probability for a specific $p\in[\frac{1}{2},1]$ and a specific CHSH violation $I_{obs}\in[2,2\sqrt{2}]$ as $P^{(q, A)}_{g}(I_{obs}, \{p,1-p\})$. Then we have 
We finally obtain
\begin{equation}\label{ave_guess_p}
	\begin{split}
		&P^{(q, A)}_{g}(I_{obs}, \{p,1-p\}):= \max_{\substack{\{p_{ABE|XYZ}\} \in \mathcal{Q}^{trip}, \\ I_{CHSH}(p_{AB|XY}) = I_{obs}}} \sum_{a= \pm 1} p \cdot p_{AE|XZ}(a, e=a|A_0, z=0) + (1-p) \cdot p_{AE|XZ}(a, e=a|A_1, z= 1)\\
		&= \max_{\substack{\{p_{ABE|XYZ} \}\in \mathcal{Q}^{trip}, \\ I_{CHSH}(p_{AB|XY}) = I_{obs}}} p \cdot \sum_{a = \pm 1} p_{E|Z}(a|z=0) p_{A|XZE}(a|A_0, z=0,e=a) \\ 
        & \qquad \qquad \qquad \qquad \qquad + (1-p) \cdot \sum_{a = \pm 1} p_{E|Z}(a|z=1)  p_{A|XZE}(a|A_1,z=1,e=a)\\
		& = \max_{\substack{\{p_{ABE|XYZ}\} \in \mathcal{Q}^{trip}, \\ I_{CHSH}(p_{AB|XY}) = I_{obs}}} p \cdot \sum_{a\in\pm 1} p_{E|Z}(a|z=0)  \cdot \frac{1+ a \<A_0\>_{|z=0,e=a}}{2} + (1-p) \cdot \sum_{a\in\pm 1} p_{E|Z}(a|z=1) \cdot  \frac{1+a \<A_1\>_{|z=1,e=a}}{2}\\
		&\geq \frac{1}{2} + \frac{1}{2} \max_{\substack{\{p_{ABE|XYZ=z^*}\} \in \mathcal{Q}^{trip}, \\ I_{CHSH}(p_{AB|XY}) = I_{obs}}} \sum_{a\in\pm 1} p_{E|Z}(a|z=z^*) \left[p \cdot a \<A_0\>_{|z=z^*,e=a}+(1-p) \cdot a \<A_1\>_{|z=z^*,e=a}\right]\\	
		&= \frac{1}{2} + \frac{\sqrt{p^2 + (1-p)^2}}{2} f_{\phi_p}(I_{obs}). \\	
	\end{split}
\end{equation}
Here, we have used $p_{A|X,Z,E}(a|A_0, z=0,e=a) = \frac{1+ a \<A_0\>_{|z=0,e=a}}{2}$ in the third line. In the fourth line, we consider that Eve performs a single fixed measurement $Z=z^*$ to guess Alice's outputs for both $A_0$ and $A_1$, which is a subset of strategies considered in previous lines, resulting in the inequality. In the final line, we have used  Eqs.~\eqref{db_tilted_bound_2} and \eqref{db_tilted_bound_3} and the fact that $\sum_{a}p_{E|Z}(a|z=z^*) = 1$. 
\end{proof}
%In the first line, $p_{A,B,E|X,Y,Z} Q$ indicates that $p_{ABE}$ is a quantum behavior, and $p_{AE}$ and $p_{AB}$ denote its Alice--Eve and Alice--Bob marginals. The condition $I_{CHSH}(p_{AB})=I_{obs}$ specifies the CHSH violation. $E_0$ and $E_1$ are the POVMs used by Eve to guess Alice's outcomes from different settings. In the third line, $\<A_i\>_{|e=a}$ denotes the expectation value of $A_i$ conditioned on Eve's outcome $e=a$, and we use the relation $\<A_i\>=\sum_{a} (-1)^a p(a|x=i)$ to derive $p(a|x=i)=\frac{1+(-1)^a \<A_i\>}{2}$. In the fourth line, the inequality holds because we implicitly use $p_E(a|E_0)=p_E(a|E_1)=p_E(a)$ for all $a\in\{0,1\}$. In the last line, we use the tight bounds from Eqs.~\eqref{db_tilted_bound_2} and \eqref{db_tilted_bound_3}.
We plot the lower bound of the average guessing probability $P^{(q, A)}_{g}(I_{obs}, \{p,1-p\})$ (an upper bound on the average min-entropy) derived in Eq.~\eqref{ave_guess_p} for $p=0.5,0.8,0.9,1$ in Fig.~\ref{fig_pg_chsh}.

We also obtain a numerical upper bound on the average guessing probability (a lower bound on the average min entropy) from the level $3$ of the NPA hierarchy of semidefinite programs as shown in Fig.~\ref{fig_pg_chsh} (which we also certify using a feasible solution to the dual program). As evident from Fig.~\ref{fig_pg_chsh}, the analytical bound in Eq.~\eqref{ave_guess_p} is tight, the analytical lower bound coincides with the numerical upper bound up to numerical precision of $10^{-9}$. The corresponding min-entropy $H_{\min}^{(I_{obs}, p)}(a|x,E) =-\log_2\big(P^{(q, A)}_{g}(I_{obs}, \{p,1-p\})\big)$ is shown in Fig.~\ref{fig_hmin_chsh}.

\begin{figure}[htbp]
	\centering
	\subfigure[]{
		\begin{minipage}[t]{0.48\linewidth}
			\centering
			\includegraphics[width=\textwidth]{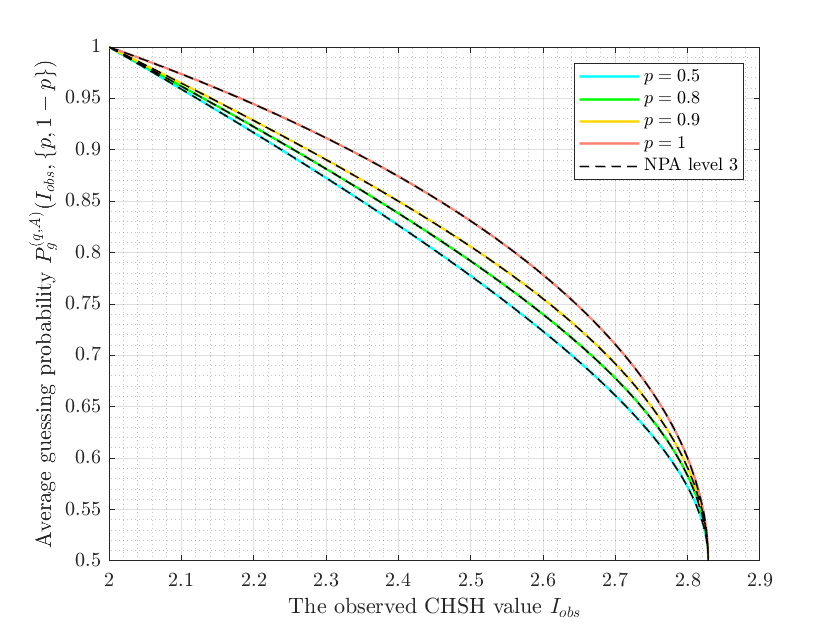}
			\label{fig_pg_chsh}
		\end{minipage}%
	}%
	\subfigure[]{
		\begin{minipage}[t]{0.48\linewidth}
			\centering
			\includegraphics[width=\textwidth]{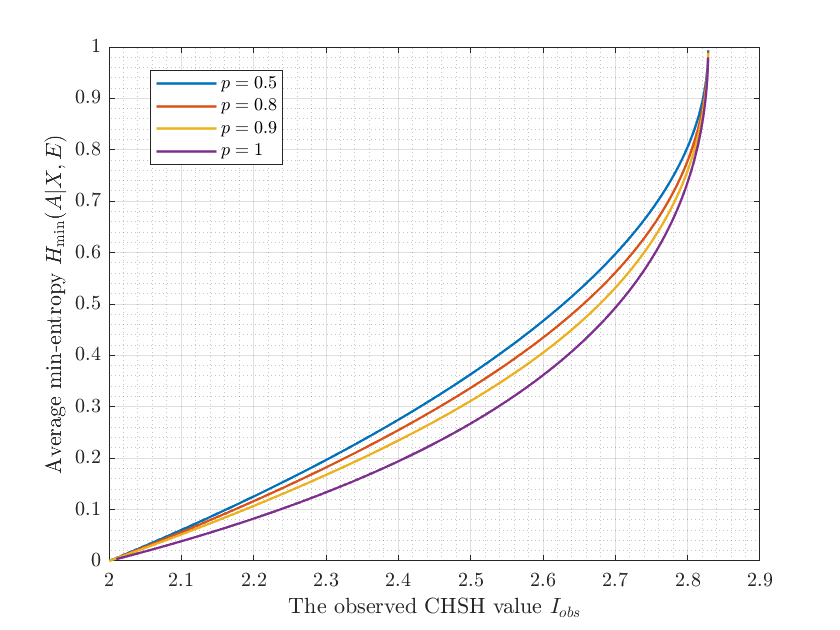}
			\label{fig_hmin_chsh}
		\end{minipage}%
	}
	\centering
	\caption{(a) Solid lines: The lower bound of the average guessing probability $P^{(q, A)}_{g}(I_{obs}, \{p,1-p\})$ as a function of the observed CHSH $I_{obs}\in[2,2\sqrt{2}]$, derived in Eq.~\eqref{ave_guess_p}, where $A_0$ is chosen with probability $p=0.5,0.8,0.9,1$ and $A_1$ is chosen with probability $1-p$. Dashed lines: the upper bounds on the average guessing probability obtained from the level-3 NPA hierarchy~\cite{NPA1,NPA2}. 
	(b) The average min-entropy $H_{\min}(A|X,E)$ as a function of the observed CHSH value $I_{obs}\in[2,2\sqrt{2}]$.}	
\end{figure}

Finally, we note that we recover the well-known analytical expression for the fixed-setting guessing probability in the CHSH Bell test, given as \cite{PAM+10}.
When $p=1$, we have that $\phi_{p=1}=0$. Thus for any any $I_{obs}\in[2,2\sqrt{2}]$, $\alpha_{\min, \phi_{p=1}}$ from Eq.~\eqref{eq_h_phi} and $I_Q^{(\alpha_{\min,\phi_{p=1}},\phi_{p=1})}$ from eq.~\eqref{eq_Iq_dtchsh} are calculated as
\begin{equation}
    \begin{split}
        \alpha_{\min, \phi_{p=1}}=\frac{2\sqrt{8-I_{obs}^2}}{I_{obs}}, \qquad I_Q^{(\alpha_{\min,\phi_{p=1}},\phi_{p=1})}=\frac{8}{I_{obs}}.
    \end{split}
\end{equation}
Thus we obtain the function $f_{\phi_{p=1}}$ and the corresponding guessing probability for $p=1$ and any $I_{obs}\in[2,2\sqrt{2}]$ as
\begin{eqnarray}
    f_{\phi_{p=1}}=\frac{\sqrt{8-I_{obs}^2}}{2},\qquad 
    P^{(q, A)}_{g}(I_{obs}, \{1,0\}) = \frac{1}{2} + \frac{1}{2} \sqrt{2-\frac{I
    ^2_{obs}}{4}}.
\end{eqnarray}

\section{Improving the randomness generation rate in DIRA against quantum adversaries}
\label{app:DIRA-imp}

In this section, we apply the scenario considered in Theorem~\ref{thm:DIRA-main}, where Eve attempts to guess Alice’s outcome averaged over all inputs $x$ rather than a fixed input $x^*$, to improve the randomness generation rate of the state-of-the-art DIRA protocol~\cite{KAF20}. In the main text, we discuss the average guessing probability, whose negative logarithm corresponds to the min-entropy. However, the DIRA protocol of~\cite{KAF20} quantifies randomness using the von Neumann entropy of Alice’s outputs conditioned on Eve’s information. To enable a fair comparison and to demonstrate possible improvements, we therefore extend our framework to compute the average von Neumann entropy over all inputs, subject to the SV-source~\cite{SV84} distributional constraints consistent with the protocol settings.  

Suppose that Alice and Bob’s measurement inputs are independently generated from a structured SV-source~\cite{SV84}, where each input bit is $\epsilon$-biased from uniform given all previous bits. Explicitly, in the $i$-th round of the protocol, Alice’s input distribution $\mu_{X}(x)$ satisfies
\begin{equation}\label{sv_con}
	\frac{1}{2}-\epsilon\leq \mu_{X_i|X_1,\ldots,X_{i-1}}(x_i|x_1,\ldots,x_{i-1})\leq \frac{1}{2}+\epsilon
\end{equation}
and similarly for Bob's inputs $y_i$.  
Based on this input distribution, denoting $\epsilon_{\min}:=(\frac{1}{2}-\epsilon)^2$ and $\epsilon_{\max}:=(\frac{1}{2}+\epsilon)^2$, the protocol employs the following bipartite binary-input and binary-output measurement-dependent-locality (MDL) Bell test:
\begin{equation}
	I^{(\epsilon)}:= \epsilon_{\min}\cdot p_{ABXY}(0,0,0,0)
	-\epsilon_{\max}\cdot \big(p_{ABXY}(0,1,0,1) + p_{ABXY}(1,0,1,0) + p_{ABXY}(0,0,1,1) \big).
\end{equation}
The maximal classical value for this MDL Bell test is $0$ for any $\epsilon\in[0,1/2]$. However, this bound can always be violated for any $\epsilon\in(0,1/2]$ by quantum strategies that violate the Hardy paradox~\cite{Hardy93} regardless of the input distribution $\mu_{X,Y}$.

In~\cite{KAF20}, the authors derive the relation between the von Neumann entropy of the output for the fixed input $x^*=0$ and the observed value $I_{obs}^{(\epsilon)}$ of the MDL Bell test, denoted as $H^{(A,\epsilon)}(I_{obs}^{(\epsilon)},x^*=0)$. To achieve this, it is useful to bound the value of the MDL Bell test by a Bell test that does not depend on the input distribution:
\begin{equation}
	\begin{split}
		I^{(\epsilon)}&=\epsilon_{\min}\cdot p_{AB|XY}(0,0|0,0) \cdot \mu_{XY}(0,0) \\
		&- \epsilon_{\max}\cdot \Big( p_{AB|XY}(0,1|0,1) \cdot \mu_{XY}(0,1)
		+p_{AB|XY}(1,0|1,0)\cdot \mu_{XY}(1,0)
		+p_{AB|XY}(0,0|1,1)\cdot \mu_{XY}(1,1) \Big)\\
		&\leq \epsilon_{\min}\epsilon_{\max} \cdot \left( p_{AB|XY}(0,0|0,0)-  p_{AB|XY}(0,1|0,1) - p_{AB|XY}(1,0|1,0)-p_{AB|XY}(0,0|1,1)\right)\\
		&=\epsilon_{\min}\epsilon_{\max} \widetilde{I}.\\
	\end{split}
\end{equation}
where we denote $\widetilde{I}:= p_{AB|XY}(0,0|0,0)-  p_{AB|XY}(0,1|0,1) - p_{AB|XY}(1,0|1,0)-p_{AB|XY}(0,0|1,1)$, a Bell expression that is independent of the input distribution. The quantum value of $\widetilde{I}$ lies in the interval $[0,\frac{\sqrt{2}-1}{2}]$. In the second line of the above derivation, we use the fact that $\epsilon_{\min}\leq \mu_{XY}(x,y)\leq \epsilon_{\max},\forall xy\in\{00,01,10,11\}$ and the non-negativity of $p_{AB|XY}$. One should note that, for a given $\epsilon\in(0,\tfrac{1}{2}]$, although $I^{(\epsilon)}$ admits the upper bound $\epsilon_{\min}\epsilon_{\max}\,\widetilde{I}$, the observed value $I^{(\epsilon)}_{obs}$ in the MDL Bell test does not necessarily reach $\epsilon_{\min}\epsilon_{\max}\,\tfrac{\sqrt{2}-1}{2}$. Nevertheless, we can always derive a lower bound on $I^{(\epsilon)}$:
\begin{equation}
	I^{(\epsilon)} \geq \epsilon_{\min}^2 \cdot p_{AB|XY}(0,0|0,0)
	- \epsilon_{\max}^2 \cdot \Big( p_{AB|XY}(0,1|0,1)
	+ p_{AB|XY}(1,0|1,0)
	+ p_{AB|XY}(0,0|1,1) \Big).
\end{equation}
Denoting the optimal quantum value of the right-hand side as $I_{Q,\text{low}}^{(\epsilon)}$, the observed value $I^{(\epsilon)}_{obs}$ of the MDL Bell test must lie within $[0, I_{Q,\text{low}}^{(\epsilon)}]$.

To demonstrate the improvement in randomness generation rate compared to~\cite{KAF20}, we compute the von Neumann entropy averaged over Alice’s input distribution $\mu_X$, subject to the SV-source constraints in Eq.~\eqref{sv_con}, given the observed MDL Bell test value $I_{obs}^{(\epsilon)}$. In other words, we are interested in solving the following optimization problem:
\begin{equation}\label{DIRA_opt1}
	\begin{split}
		\min_{\mu_{X},\; p_{ABE|XYZ}} & \sum_{x\in\{0,1\}} \mu_{X}(x) H(A|x,E,Z)\\
		\text{s.t.} \qquad \quad & \widetilde{I}(p_{AB|XY})\geq \frac{I_{obs}^{(\epsilon)}}{\epsilon_{\min}\epsilon_{\max}}\in [0,\frac{I_{Q,\text{low}}^{(\epsilon)}}{\epsilon_{\min}\epsilon_{\max}}], \; \; \forall z\\
		& \tfrac{1}{2}-\epsilon \leq \mu_{X}(x)\leq \tfrac{1}{2}+\epsilon \; \; \; \; \forall x, \qquad \sum_{x} \mu_{X}(x)=1;\\
		& \{p_{ABE|XYZ}\}\in\mathcal{Q}^{trip}.\\
	\end{split}
\end{equation}

Here $H(A|x,E,Z)$ denotes the von Neumann entropy of Alice’s outcome conditioned on Eve’s quantum side information $E,Z$, given that Alice’s input is $x$. Denote the solution of the above optimization problem as $H^{(A,\epsilon)}(I_{obs}^{(\epsilon)},\text{SV}_{\epsilon})$. 
An important observation is that the optimal value of~\eqref{DIRA_opt1} can always be achieved when $\mu_{X}(x=0)=\tfrac{1}{2}-\epsilon,\mu_{X}(x=1)=\tfrac{1}{2}+\epsilon$ or $\mu_{X}(x=0)=\tfrac{1}{2}+\epsilon,\mu_{X}(x=1)=\tfrac{1}{2}-\epsilon$. This is because, for any fixed tripartite quantum behavior $p_{A,B,E|X,Y,Z}$, the values of $H(A|x=0,E,Z)$ and $H(A|x=1,E,Z)$ are fixed. Without loss of generality, assume $H(A|x=0,E,Z)\geq H(A|x=1,E,Z)$. In this case, the objective function is minimized when the input probability $\mu_X(x=0)$ takes its lowest allowed value, $\tfrac{1}{2}-\epsilon$, with $\mu_X(x=1)$ taking its biggest allowed value $\tfrac{1}{2}+\epsilon$. 

Using this observation, the optimization problem reduces to:
\begin{equation}
	\begin{split}
		\min_{\{p_{ABE|XYZ}\}} & \left(\tfrac{1}{2}-\epsilon\right) H(A|x=0,E,Z)+\left(\tfrac{1}{2}+\epsilon\right) H(A|x=1,E,Z)\\
		\text{s.t.} \qquad & \widetilde{I}(p_{AB|XY})\geq \tfrac{I_{obs}^{(\epsilon)}}{\epsilon_{\min}\epsilon_{\max}}\in [0,\frac{I_{Q,\text{low}}^{(\epsilon)}}{\epsilon_{\min}\epsilon_{\max}}], \;\; \forall z;\\
		& \{p_{ABE|XYZ} \}\in\mathcal{Q}^{trip}.\\
	\end{split}
\end{equation}
This optimization problem can be efficiently solved by adapting the semidefinite programming method introduced by Brown, Fawzi, and Fawzi in~\cite{BFF24}. To illustrate the improvement, Figure \ref{fig_DIRA} compares two von Neumann entropy curves, $H^{(A,\epsilon)}(I_{obs}^{(\epsilon)},x^*=0)$, the von Neumann entropy at the fixed input $x^*=0$, and $H^{(A,\epsilon)}(I_{obs}^{(\epsilon)},\text{SV}_{\epsilon})$,  the average von Neumann entropy. The comparison is performed for five different values of $\epsilon\in\{0,0.0125,0.025,0.0375,0.05\}$. All curves are obtained using the same parameters as explained in \cite{BFF24}, with the NPA hierarchy level set to $2+ABZ$ and the Gauss–Radau quadrature level set to $6$.

\begin{figure}
    \centering
    \includegraphics[width=0.7\linewidth]{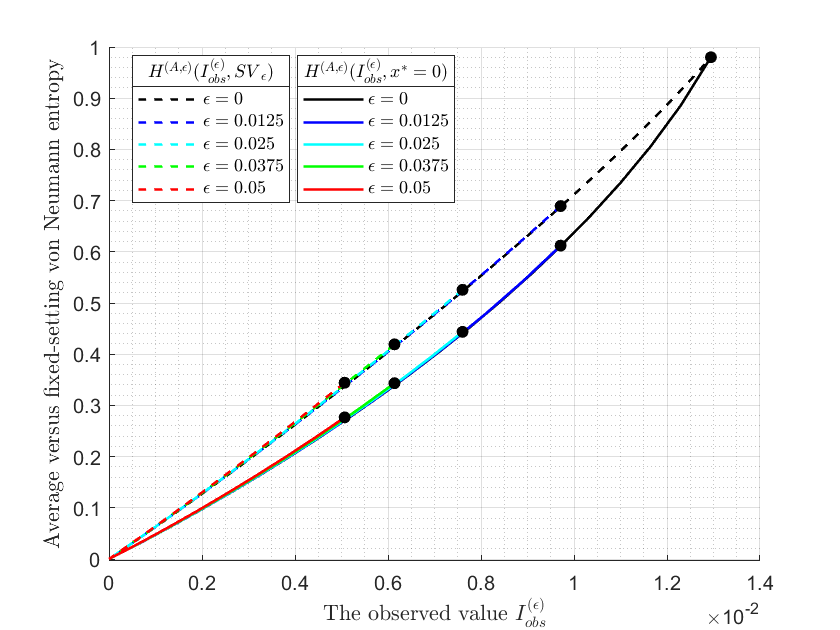}
    \caption{Comparison between $H^{(A,\epsilon)}(I_{obs}^{(\epsilon)},x^*=0)$, the von Neumann entropy at fixed input $x^*=0$, and $H^{(A,\epsilon)}(I_{obs}^{(\epsilon)},\text{SV}_{\epsilon})$, the average von Neumann entropy. Both quantities are obtained using the adapted SDP method of Brown, Fawzi, and Fawzi ~\cite{BFF24}, with the NPA hierarchy~\cite{NPA1,NPA2} level set to $2+ABZ$ and Gauss–Radau quadrature level set to $6$.}
    \label{fig_DIRA}
\end{figure}

\section{Analytical comparison of fixed-setting vs average guessing probability for an optimized Bell expression}
\label{app:Pgavg-simplex}
In this Appendix, we present a new bipartite Bell inequality involving three binary-outcome measurement settings for each party, and we provide an explicit analytical characterization of the optimal guessing strategy for a quantum eavesdrop Eve. Specifically, we analytically calculate the adversary’s guessing probability for Alice’s outcome, both in the fixed-setting scenario and in terms of the average guessing probability. The feature of this Bell expression, as we shall see, is that it allows for perfect correlations between Alice and Bob's measurement outcomes for pairs $(A_i, B_i)$ for every $i$ and thus potentially provides an improvement over the CHSH expression in tasks such as key distribution, when multiple settings are chosen during the generation rounds, a feature which we explore in forthcoming work.  

Denote by $A_0, A_1, A_2$ the three binary-outcome measurements (with outcomes $\{\pm 1\}$) on Alice's subsystem, and by $B_0, B_1, B_2$ the three binary-outcome measurements on Bob's subsystem. The Bell expression is defined as
\begin{equation}\label{Id}
	I = \sum_{i=0}^2 \<A_iB_i\> - \sum_{\substack{i,j=0,\\ i \neq j}}^2 \<A_iB_j\>,
\end{equation}
where the correlators $\<A_i B_j\>$ are defined as $\langle A_i B_j \rangle:= \sum _{a,b\in\{\pm 1\}} ab\cdot p_{AB|XY}(a,b|A_i,B_j).$
It is straightforward to verify that the maximum local value of this Bell expression is $I_{L}=5$. On the other hand, the maximum quantum value is $I_{Q}=6$, as established in the following lemma.

\begin{lem}
The maximum quantum value of the Bell expression in Eq.~\eqref{Id} is $I_{Q}=6$.
\end{lem}

\begin{proof}
We first show that the lower bound of $I_Q$ is $6$ by exhibiting a quantum strategy that achieves this value. Consider the maximally entangled two-qubit state $|\psi^+\> = \frac{1}{\sqrt{2}}(|00\> + |11\>)$ shared between Alice and Bob. They perform the following synchronous measurements:
\begin{equation}\label{sync_operator}
	\begin{split}
		A_0 = B_0 &= \sigma_z,\\
		A_1 = B_1 &= \cos{\frac{4\pi}{3}}\,\sigma_z + \sin{\frac{4\pi}{3}}\,\sigma_x,\\
		A_2 = B_2 &= \cos{\frac{2\pi}{3}}\,\sigma_z + \sin{\frac{2\pi}{3}}\,\sigma_x,
	\end{split}
\end{equation}
where $\sigma_z$ and $\sigma_x$ denote the Pauli operators. This strategy achieves value $6$, showing that $I_Q \geq 6$.

Next, we prove that the upper bound of $I_Q$ is also $6$. To this end, consider the sum-of-squares (SOS) decomposition of the shifted Bell operator $\overline{I} := 6\mathbb{I} - I$. Taking the monomial basis $R = [A_0, A_1, A_2, B_0, B_1, B_2]^{\dagger}$, the shifted Bell operator $\overline{I}$ can be written as
\begin{equation}
	\overline{I} = R^{\dagger} M R, \quad \text{with the PSD matrix } M = \left(\begin{matrix}
		1 & 0 & 0 & -\frac{1}{2} & \frac{1}{2} & \frac{1}{2}\\
		0 & 1 & 0 & \frac{1}{2} & -\frac{1}{2} & \frac{1}{2}\\
		0 & 0 & 1 & \frac{1}{2} & \frac{1}{2} & -\frac{1}{2}\\
		-\frac{1}{2} & \frac{1}{2} & \frac{1}{2} & 1 & 0 & 0\\
		\frac{1}{2} & -\frac{1}{2} & \frac{1}{2} & 0 & 1 & 0\\
		\frac{1}{2} & \frac{1}{2} & -\frac{1}{2} & 0 & 0 & 1\\
	\end{matrix}\right).
\end{equation}
Applying the Cholesky decomposition to $M$ yields a SOS decomposition of the shifted Bell operator:
\begin{equation}
	\overline{I} = \left(A_0 - \frac{B_0}{2} + \frac{B_1}{2} + \frac{B_2}{2}\right)^2 + \left(A_1 + \frac{B_0}{2} - \frac{B_1}{2} + \frac{B_2}{2}\right)^2 + \left(A_2 + \frac{B_0}{2} + \frac{B_1}{2} - \frac{B_2}{2}\right)^2 + \left(\frac{B_0}{2} + \frac{B_1}{2} + \frac{B_2}{2}\right)^2.
\end{equation}
This shows that $\overline{I}$ is positive semidefinite $\overline{I}\succeq 0$, and hence the value of $I$ is smaller than or equal to $6$ for all quantum strategies. Combining both bounds, we conclude that $I_Q = 6$.
\end{proof}

\subsection{The Optimal Guessing Strategy in the Fixed-Setting Scenario}\label{sec_simplex_fix}
We now turn to the adversarial scenario and establish Eve’s optimal guessing strategy in the fixed-setting case. In addition, we assume that Eve attempts to guess the outcome of Alice's measurement $A_0$.
One candidate guessing strategy for Eve is as follows. Eve prepares a pure three-qubit state $|\Psi\>_{ABE}$ shared between the Alice-Bob devices and her own system:
\begin{equation}
	|\Psi\>_{ABE} = \frac{1}{\sqrt{2}}\left(|0\>_A |0\>_B |e_0\>_E + |1\>_A |1\>_B |e_1\>_E\right),
\end{equation}
where $|e_0\>$ and $|e_1\>$ are normalized qubit states with overlap $\<e_0|e_1\> = \sin{2\theta}$ for a parameter $\theta \in [\frac{\pi}{12}, \frac{\pi}{4}]$. In addition, Eve prepare the local measurement operators for Alice and Bob's subsystem:
\begin{equation}
	\begin{array}{ll}
		& A_0 = B_0 = \sigma_z, \\[0.5ex]
		& A_1 = B_1 = \frac{\sqrt{4\sin^2 2\theta - 1}}{2\sin 2\theta} \,\sigma_x - \frac{1}{2\sin 2\theta} \,\sigma_z,\\[0.5ex]
		& A_2 = B_2 = -\frac{\sqrt{4\sin^2 2\theta - 1}}{2\sin 2\theta} \,\sigma_x - \frac{1}{2\sin 2\theta} \,\sigma_z.
	\end{array}
\end{equation}
Under this setup, the reduced state on Alice and Bob's subsystems is
\begin{equation}
	\rho_{AB} = \frac{\mathbb{I} \otimes \mathbb{I}}{4} + \frac{\sigma_z \otimes \sigma_z}{4} + \sin{2\theta} \,\frac{\sigma_x \otimes \sigma_x}{4} - \sin{2\theta} \,\frac{\sigma_y \otimes \sigma_y}{4}.
\end{equation}
A direct calculation then yields the observed value $I_{obs}$ of the Bell expression in Eq.~\eqref{Id} under this setup:
\begin{equation}\label{wq}
	I_{obs} = \Tr\left[\rho_{AB} \left(\sum_{i=0}^2 A_i B_i - \sum_{\substack{i,j=0\\ i \neq j}}^2 A_i B_j \right)\right] = 1 + 4\sin{2\theta} + \frac{1}{\sin{2\theta}}.
\end{equation}

After Alice performs the measurement $A_0$, the correlation between the classical outcomes and Eve's system is described by the classical–quantum state:
\begin{equation}\label{cq_state1}
	\frac{1}{2} [+1]_{A_0} \otimes |e_0\>\<e_0| + \frac{1}{2} [-1]_{A_0} \otimes |e_1\>\<e_1|,
\end{equation}
where $[+1]_{A_i}$ and $[-1]_{A_i}$ label the two possible outcomes of $A_i$. Eve's optimal guessing strategy in this case is to perform the measurement that optimally discriminates between the two pure states $|e_0\>$ and $|e_1\>$ with uniform prior probabilities. According to the seminal result by Helstrom~\cite{Helstrom69}, the minimum-error discrimination is achieved by performing the two-outcome POVM $E_0 = \{E_0^0,\, E_0^1 = \mathbb{I}_2 - E_0^0\}$, where $E_0^0$ is the projector onto the positive eigenspace of the operator $\frac{1}{2} |e_0\>\<e_0| - \frac{1}{2} |e_1\>\<e_1|$. The resulting guessing probability is:
\begin{equation}\label{pg1}
	\widetilde{P}_{g,x^*=0}^{(A,q)}=\frac{1}{2}\left(1 + \sqrt{1 - |\<e_0|e_1\>|^2}\right) = \frac{1}{2}\left(1 + \cos{2\theta}\right).
\end{equation}

Combining Eq.~\eqref{pg1} with the observed quantum value $I_{obs}$ in Eq.~\eqref{wq}, we obtain an explicit relation between the guessing probability and $I_{obs}$. Eliminating the parameter $\theta$ yields the guessing probability as a function of $I_{obs}$, denote it by $f_{\text{fix}} (I_{obs})$
\begin{equation}\label{pg_w_1}
	\begin{split}
		f_{\text{fix}}(I_{obs}) &= \frac{1}{2} + \frac{1}{2} \sqrt{1 - \frac{1}{64} \left( I_{obs} - 1 + \sqrt{(I_{obs} + 3)(I_{obs}- 5)} \right)^2 }.
	\end{split}
\end{equation}

However, taking the second derivative of $f_{\text{fix}}(I_{obs})$ shows that it is not concave. Furthermore, when $I_{obs} = 5$, corresponding to the maximum classical value, the system behaves classically and Eve can guess the outcome of $A_0$ with certainty, i.e., guessing probability equals $1$. However, $f_{\text{fix}}(5) = \frac{1}{2} + \frac{\sqrt{3}}{4}$, which is strictly less than $1$. This implies that for non-maximal quantum value of the Bell expression, Eve may have more effective strategies to guess Alice's outcome.

One such strategy is to prepare a convex mixture of a local behavior and the quantum behavior that achieves the value $I_{\text{fix}}^*$, thereby generating a quantum value for $I$ smaller than $I_{\text{fix}}^*$. In this case, the guessing probability becomes the convex combination of $1$ and $f_{\text{fix}}(I_{\text{fix}}^*)$ with corresponding weights. The concrete value of $I_{\text{fix}}^*$ is the maximizer of the following optimization problem:
\begin{equation}
	\max_{I_{obs}} \quad \frac{f_{\text{fix}}(I_{obs}) - 1}{I_{obs}- 5} \quad \rightarrow \quad I_{\text{fix}}^* = \frac{1}{3}(4 + 5\sqrt{7}).
\end{equation}
Note that $f_{\text{fix}} (I_{\text{fix}}^*) = \frac{1}{8}(3 + \sqrt{7})$. Thus, under this strategy, the guessing probability takes the form of a piecewise function:
\begin{equation}\label{pg_0}
	P_g^{(A,q)}(I_{obs},x^*=0) =
	\begin{cases}
		\frac{1}{2} + \frac{1}{2} \sqrt{1 - \frac{1}{64} \left( I_{obs} - 1 + \sqrt{(I_{obs} + 3)(I_{obs}- 5)} \right)^2}, & I_{obs} \in \left[ \frac{1}{3}(4 + 5\sqrt{7}),\, 6 \right], \\[1ex]
		-\frac{10 + 7\sqrt{7}}{72} I_{obs} + \frac{122 + 35\sqrt{7}}{72}, & I_{obs} \in \left[ 5,\, \frac{1}{3}(4 + 5\sqrt{7}) \right].
	\end{cases}
\end{equation}

A numerical upper bound on the guessing probability can be obtained using the Navascués-Pironio-Acín hierarchy of SDPs~\cite{NPA1,NPA2}. Figure~\ref{fig_pg} displays the analytical curve of Eq.~\eqref{pg_0}, derived from our candidate strategy, alongside the SDP upper bound obtained from NPA level $k=3$. As evident from Figure~\ref{fig_pg}, the analytical curve in Eq.~\eqref{pg_0} is tight, confirming the optimality of our candidate strategy. The corresponding min-entropy $H_{\min}(A|x^*=0,E) = -\log_2(P_g^{(A,q)}(I_{obs},x^*=0))$ is shown in Figure~\ref{fig_hmin}.

\begin{figure}[htbp]
\centering
\subfigure[]{
\begin{minipage}[t]{0.48\linewidth}
\centering
\includegraphics[width=\textwidth]{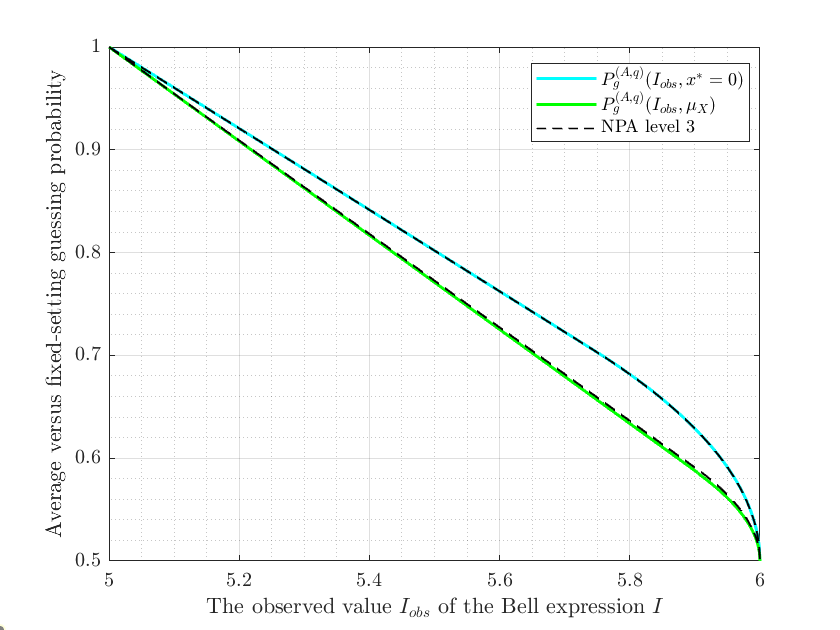}
\label{fig_pg}
\end{minipage}%
}%
\subfigure[]{
\begin{minipage}[t]{0.48\linewidth}
\centering
\includegraphics[width=\textwidth]{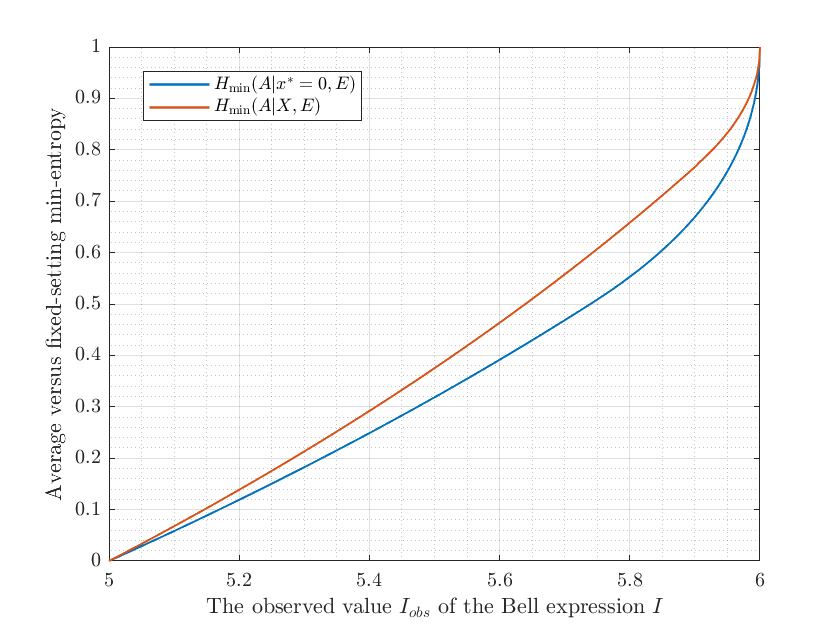}
\label{fig_hmin}
\end{minipage}%
}
\centering
\caption{(a) Solid lines: The guessing probability $P_g^{(A,q)}(I_{obs},x^*=0)$ of Eq.~\eqref{pg_0} and the average guessing probability ${P}_g^{(A,q)}(I_{obs},\mu_X)$ of Eq.~\eqref{pg_x} for $\mu_X(x)=\frac{1}{3},\forall x\in\{0,1,2\}$ as functions of the observed quantum value $I_{obs}$ of the Bell expression $I$ Eq.~\eqref{Id}, obtained from the candidate guessing strategies. Dashed lines: the upper bounds on the guessing probability and average guessing probability derived from the level-3 NPA hierarchy~\cite{NPA1,NPA2}. 
(b) The min-entropy $H_{\min}(A|x^*=0,E)$ and the average min-entropy $H_{\min}(A|X,E)$ as functions of the observed quantum value $I_{obs}$ of the Bell expression $I$ Eq.~\eqref{Id}.}
\end{figure}

\subsection{The Optimal Average Guessing Probability}
We now turn to the case of the average guessing probability, where Eve attempts to guess Alice’s outcomes across all three measurement settings $A_0, A_1, A_2$, which are chosen uniformly with distribution $\mu_X(x)=\frac{1}{3},\forall x\in\{0,1,2\}$. We assume that Eve prepares the same shared state and measurements as in the previous subsection~\ref{sec_simplex_fix}. As explained earlier, when Alice measures $A_0$, the correlation between her outcome and Eve’s system is described by Eq.~\eqref{cq_state1}, and Eve applies the two-outcome POVM $E_0$ on her system to guess Alice’s result. Now, when Alice measures $A_1$, the correlation between her classical outcome and Eve’s system is described by the following classical–quantum state:
\begin{equation}\label{cq_state2}
\begin{split}
	&\frac{1}{2} [+1]_{A_1} \otimes \left[ \left( \frac{1}{2} + \frac{1}{4\sin 2\theta} \right) |e_0\>\<e_0| + \left( \frac{1}{2} - \frac{1}{4\sin 2\theta} \right) |e_1\>\<e_1| \right]\\
	+ &\frac{1}{2} [-1]_{A_1} \otimes \left[ \left( \frac{1}{2} - \frac{1}{4\sin 2\theta} \right) |e_0\>\<e_0| + \left( \frac{1}{2} + \frac{1}{4\sin 2\theta} \right) |e_1\>\<e_1| \right].
\end{split}
\end{equation}
We define
\begin{equation}
	\begin{split}
		\sigma_0 &:= \left( \frac{1}{2} + \frac{1}{4\sin 2\theta} \right) |e_0\>\<e_0| + \left( \frac{1}{2} - \frac{1}{4\sin 2\theta} \right) |e_1\>\<e_1|,\\
		\sigma_1 &:= \left( \frac{1}{2} - \frac{1}{4\sin 2\theta} \right) |e_0\>\<e_0| + \left( \frac{1}{2} + \frac{1}{4\sin 2\theta} \right) |e_1\>\<e_1|.
	\end{split}
\end{equation}

Again, Eve's optimal guessing strategy in this case is to perform the measurement that optimally discriminates between the two mixed states $\sigma_0$ and $\sigma_1$ with uniform prior probabilities. The minimum-error discrimination is achieved by the two-outcome POVM $E_1 = \{E_{1}^{0},\,E_1^{1}\}$, where $E_1^0$ ($E_1^{1}$) is the projector onto the positive (negative) eigenspace of the operator $\frac{1}{2}\sigma_0 - \frac{1}{2}\sigma_1$. The resulting guessing probability under this strategy is 
\begin{equation}
	\widetilde{P}_{g,x^*=1}^{(A,q)} = \frac{1}{2} \left( 1 + \left \| \frac{1}{2} \sigma_0 - \frac{1}{2} \sigma_1 \right\|_{tr} \right) = \frac{1}{2} \left( 1 + \frac{\cos 2\theta}{2 \sin 2\theta} \right).
\end{equation}
Similarly, when Alice measures $A_2$, the classical–quantum state is equivalent to Eq.~\eqref{cq_state2}, and thus $\widetilde{P}_{g,x^*=2}^{(A,q)} = \widetilde{P}_{g,x^*=1}^{(A,q)}$. Thus the average guessing probability among all three settings is:
\begin{equation}
	\widetilde{P}_{g,\mu_X}^{(A,q)} = \frac{1}{3} \widetilde{P}_{g,x^*=0}^{(A,q)} + \frac{2}{3}\widetilde{P}_{g,x^*=1}^{(A,q)} = \frac{1}{2} + \frac{1}{6} \left( \cos 2\theta + \frac{\cos 2\theta}{\sin 2\theta} \right).
\end{equation}
Eliminating the parameter $\theta$ in $\widetilde{P}_{g,\mu_X}^{(A,q)}$ and the observed quantum value $I_{obs}$ from Eq.~\eqref{wq} yields the average guessing probability as a function of $I_{obs}$, denote it by $f_{\text{ave}} (I_{obs})$
\begin{equation}\label{pg_w_ave}
	\begin{split}
		f_{\text{ave}} (I_{obs}) &= \frac{1}{2} + \frac{1}{6} \left( 1 + \frac{8}{I_{obs} - 1 + \sqrt{(I_{obs} + 3)(I_{obs} - 5)}} \right) \sqrt{1 - \frac{1}{64} \left( I_{obs} - 1 + \sqrt{(I_{obs} + 3)(I_{obs} - 5)} \right)^2 }.
	\end{split}
\end{equation}

Similar to the previous case, since this function is not concave and $f_{\text{ave}}(5) = \frac{1}{2} + \frac{\sqrt{3}}{4} < 1$ when $I_{obs} = 5$ (the classical upper bound), which implies that this guessing strategy is not optimal for non-maximum quantum value of the Bell expression. One possible strategy for Eve is to prepare a convex mixture of a local behavior and the quantum behavior that achieves the value $I_{\text{ave}}^{*}$ to obtain quantum value smaller than $I_{\text{ave}}^{*}$. In this circumstance, the guessing probability is the convex hull of $1$ and $f_{\text{ave}}(I_{\text{ave}}^{*})$ with corresponding weights. The value $I_{\text{ave}}^{*}$ is determined by the maximizer of the following optimization problem:
\begin{equation}
	\max_{I_{obs}}\quad \frac{f_{\text{ave}} (I_{obs})-1}{I_{obs}-5} \quad \rightarrow \quad I_{\text{ave}}^{*}=1+4\frac{\left(3+2\sqrt{6}\cos(\frac{1}{3}\arccos\sqrt{\frac{2}{3}})\right)^2-1}{\left(3+2\sqrt{6}\cos(\frac{1}{3}\arccos\sqrt{\frac{2}{3}})\right)^2+1}+\frac{\left(3+2\sqrt{6}\cos(\frac{1}{3}\arccos\sqrt{\frac{2}{3}})\right)^2+1}{\left(3+2\sqrt{6}\cos(\frac{1}{3}\arccos\sqrt{\frac{2}{3}})\right)^2-1}.
\end{equation}

Thus, by using this strategy, the average guessing probability is a piecewise function:
\begin{equation}\label{pg_x}
	{P}_g^{(A,q)}(I_{obs},\mu_X) =
	\begin{cases}
		f_{\text{ave}} (I_{obs}), & I_{obs} \in [I_{\text{ave}}^{*}, 6],\\
		\frac{f_{\text{ave}}(I_{\text{ave}}^{*}) - 1}{I_{\text{ave}}^{*} - 5} (I_{obs} - 5) + 1, & I_{obs} \in [5, I_{\text{ave}}^{*}].
	\end{cases}
\end{equation}

Figure~\ref{fig_pg} displays the analytical curve of Eq.~\eqref{pg_x}, derived from our candidate strategy, alongside the SDP upper bound obtained from the level $3$ of the NPA hierarchy. As shown in Figure~\ref{fig_pg}, the analytical curve in Eq.~\eqref{pg_x} is tight, confirming the optimality of our candidate strategy. The corresponding min-entropy $H_{\min}^{I_{obs}, \mu_X}(A|X,E) = -\log_2({P}_g^{(A,q)}(I_{obs},\mu_X))$ is illustrated in Figure~\ref{fig_hmin}.

%%%%%%%%%%%%%%%%%%%%%%%%%%%%%%%%%%%%%%%%%%%%%%%%%%%%%%%%%%%%%%%%%%%

\end{document}